\documentclass[]{IEEEtran}
\usepackage{times,amsmath,amssymb,graphicx,epsfig, color, algorithmic, algorithm}

\usepackage[normalem]{ulem}

\newtheorem{lemma}{Lemma}
\newtheorem{proposition}{Proposition}
\newtheorem{assumption}{Assumption}

\DeclareMathOperator*{\argmin}{arg\,min}
\DeclareMathOperator*{\argmax}{arg\,max}

\DeclareMathOperator{\SNR}{SNR}

\def\mb#1{\mathbf{#1}}

\def\nn{\nonumber}
\def\beq{\begin{equation}}
\def\eeq{\end{equation}}
\def\beqa{\begin{eqnarray}}
\def\eeqa{\end{eqnarray}}
\def\ie{{\it i.e.,\ \/}}
\def\defeq{\stackrel{\Delta}{=}}

\def\Pvec{\boldsymbol{\pi}}
\def\pmnkvec{\boldsymbol{p}_{mnk}}
\def\lvec{\boldsymbol{\lambda}}
\def\thetavec{\boldsymbol{\theta}}

\def\l{\lambda}
\def\a{\alpha}
\def\SNR{\text{SNR}}
\def\nom{\text{nom}}

\allowdisplaybreaks

\begin{document}

\title{Jointly Optimal Channel and Power Assignment for Dual-Hop Multi-channel Multi-user Relaying}

\author{Mahdi Hajiaghayi, \IEEEmembership{Student Member, IEEE}, Min Dong, \IEEEmembership{Senior Member, IEEE}, and Ben Liang, \IEEEmembership{Senior Member, IEEE}
\thanks{M.~Hajiaghayi and B.~Liang are with the Department of Electrical and Computer Engineering, University of Toronto, Canada. M.~Dong is with the Department of Electrical Computer and Software Engineering, University of Ontario Institute of Technology, Canada.  Emails:  \{mahdih,liang\}@comm.utoronto.ca, min.dong@uoit.ca. This work was supported in part by the Natural Sciences and Engineering Research Council of Canada and the Ontario Ministry of Research and Innovation.  A preliminary version of this work has appeared in \cite{Hajiaghayi-Info11}.  This is the full version of a paper to appear in the \textit{IEEE Journal on Selected Areas in Communications, Special Issue on Cooperative Networking –- Challenges and Applications (Part II)}, October 2012.
}
}

\markboth{\protect \footnotesize Long version of IEEE J.~Select.~Areas Commun., vol.~30, no.~9, October 2012}{Hajiaghayi, Dong, and Liang}

\maketitle


\begin{abstract}

We consider the problem of jointly optimizing channel pairing, channel-user assignment, and power allocation, to maximize the weighted sum-rate, in a single-relay cooperative system with multiple channels and multiple users. Common relaying strategies are considered, and transmission power constraints are imposed on both individual transmitters and the aggregate over all transmitters. The joint optimization problem naturally leads to a mixed-integer program.  Despite the general expectation that such problems are intractable, we construct an efficient algorithm to find an optimal solution, which incurs computational complexity that is polynomial in the number of channels and the number of users.  We further demonstrate through numerical experiments that the jointly optimal solution can significantly improve system performance over its suboptimal alternatives.

\end{abstract}

\section{Introduction}

We consider the problem of resource assignment for multi-channel multi-user communication through relaying. The problem typically arises in cellular communication or wireless local area networks, through either dedicated relay stations or users temporarily serving as relay nodes.
In traditional narrow-band cooperative relaying systems, the relay
retransmits a processed version of the received signal over the same frequency channel.
In contrast, when multiple frequency channels are available, the relay can exploit the additional frequency dimension, to process incoming signals adaptively based on the diversity in channel strength.

In narrow-band cooperative relaying systems, the relay
retransmits a processed version of the received signal over the same frequency channel.
In contrast, when multiple frequency channels are available, the relay can exploit the additional frequency dimension, to process incoming signals adaptively based on the diversity in channel strength.
\textit{Channel
pairing}, which devises a matching of incoming and outgoing subcarriers in OFDM-based relaying, was
proposed independently in \cite{Hottinen2006} and \cite{Herdin2006} for
single-user relaying%
\footnote{Since a vast majority of multi-channel relaying systems in the literature are based on OFDM, we use it as an illustrative example in this work, so that the terms ``channel'' and ``subcarrier'' are synonymous.}%
. In a multi-user communication environment, both incoming
and outgoing channels at the relay are shared among all users. A crucial problem is to determine the assignment of a subset of incoming-outgoing channel pairs to each user, which we term \textit{channel-user assignment}. Since the channel condition can vary drastically for different users, and over the same incoming and outgoing channels, judicious channel-user assignment and channel pairing can potentially lead to significant improvement in spectral efficiency.  Together with \textit{power allocation} over multiple channels at the transmitters, essential for performance optimization, these are three main resource assignment problems in multi-channel multi-user relaying.

There is strong correlation among channel pairing, channel-user assignment, and power allocation. Joint consideration of these three problems is required to achieve optimal system performance.
However, the combinatorial nature of channel pairing and assignment generally leads to
a mixed-integer programming problem, whose solution often bears prohibitive computational complexity and renders the problem intractable.
As a result, previous attempts to optimize the performance of multi-channel multi-user relaying systems through resource allocation often consider only a subset of these three problems \cite{Hottinen2007}\nocite{DongHajiaghayiLiang:submTSP2011,GuoqingLi2006,Heikkinen2011,Hammerstrom2007,Vandendorpe08,ZhangInfo08,
Wang2008,Wang2009_2,Hajiaghayi:TSP2011,Dang2010,HSU2010} - \cite{Yu2007}, or  adopt  suboptimal approaches \cite{Han2009}\nocite{Zhou2009,Awad2008,CuiLau2009,XuInfocom2010,Hua2010} - \cite{Liu:TWCOM2010}.

In this work, we consider all three resource assignment problems in a dual-hop multi-channel relaying network for multi-user communication through a single relay, under several common relaying strategies.  We show that there is an efficient method to jointly optimize channel pairing, channel-user assignment, and power
allocation in such general dual-hop relaying networks.  The proposed solution framework is built upon continuous
relaxation and Lagrange dual minimization.  Although this approach is often applied to integer programming problems \cite{IP}, it generally provides only heuristic or approximate solutions.  However, by exploring the rich structure in our problem, we show that judicious reformulation and choices of the optimization trajectory can preserve both the binary constraints and the strong duality property of the continuous version, thus enabling a jointly optimal solution.

Through reformulation, we transform the core of the original problem into a special incidence of the class of three-dimensional assignment problems, which is NP hard in general but has polynomial-time solutions -- in terms of the number of channels and users -- for our specific setting of channel pairing and channel-user assignment.  For the often studied conventional decode-and-forward (DF) relaying with a maximum weighted sum-rate objective, we further propose a divide-and-conquer algorithm for dual minimization, which guarantees that convergence to an optimal solution requires only a polynomial number of iterations in the number of channels.  This ensures the scalability of the proposed solution to large multi-channel systems.

Our proposed solution is applicable to a wide range of system configurations, accommodating both total and individual power constraints, and allowing direct source-destination links in relaying.  We show that it can be modified to work with various relaying strategies in addition to DF, including variants of compress-and-forward (CF) and amplify-and-forward (AF).  It also accommodates general concave utility functions.
Through simulation and numerical comparison, we further illustrate that there is often a large performance gap between the jointly optimal solution and the suboptimal alternatives.

The rest of the paper is organized as follows. We first provide a literature review of the related work in Section \ref{sec_relatedwork}. In Section \ref{systemModel}, we discuss the system model and formulate the joint optimization problem. For weighted sum-rate maximization with DF relaying, we describe our framework of finding the optimal solution with polynomial complexity in Section \ref{sec_totindpower}. Extension to other relaying strategies are explained in Section \ref{sec_extensions}.   Numerical studies are presented in Section \ref{sec_simulation}, and conclusions are given in Section \ref{sec_conclusion}.

\section{Related Work}\label{sec_relatedwork}

Most existing works on optimizing resource allocation for multi-channel relaying systems consider a subset of the three aforementioned problems. After channel pairing was proposed in \cite{Hottinen2006} and \cite{Herdin2006}, its optimization has been considered in several studies. In the absence of the direct source-destination link, \cite{Hottinen2007}  showed that the sorted-SNR channel pairing scheme, which matches the incoming and outgoing subcarriers according to the sorted order of their SNRs for some given power allocation, is sum-rate optimal for a single-user AF relaying OFDM system. When the direct source-destination link is available, a low complexity optimal channel pairing scheme was established in \cite{DongHajiaghayiLiang:submTSP2011} for AF relaying. In addition, it was shown that channel pairing is optimal among all unitary linear processing at the relay under a fixed gain power assumptions. However, none of these works considered optimizing power allocation.
Channel-user assignment in multi-user relaying networks, under given power allocation, was considered in \cite{GuoqingLi2006}, where the authors sought an optimal channel-user and channel-relay assignment to maximize the uplink data rate for AF and DF relaying with multiple relays. For a multi-channel network with multiple sources, single AF relay, and single destination,  \cite{Heikkinen2011} studied the problem of channel pairing and channel-user assignment. It maximizes the sum received SNR, assuming that the power allocation is given. A suboptimal solution is proposed for distributed implementation using game theory.
Finally, the problem of optimal power allocation for OFDM relaying in specific relay network setups was studied in numerous works for different relay strategies and power constraints, see for example \cite{Hammerstrom2007,Vandendorpe08,ZhangInfo08}.

Jointly optimizing channel pairing and power allocation for \textit{single-user} relaying was considered in several studies. Without the direct source-destination link, \cite{Wang2008} and \cite{Wang2009_2} considered this problem for dual-hop DF relaying in an OFDM system for total power and individual power constraints, respectively. It was shown that joint subcarrier pairing and power allocation are separable for sum-rate optimization. This separation was also established in the general multi-hop case in \cite{Hajiaghayi:TSP2011}, for both AF and DF relaying, and under either total power or individual power constraints.  With consideration for the direct source-destination link, the authors of \cite{Dang2010} and \cite{HSU2010} studied joint subcarrier pairing and power allocation in a single-user OFDM system, for AF and DF relaying respectively.  The joint optimization problems were formulated as mixed-integer programs and solved in the Lagrange dual domain.  Although strict optimality was not established, the proposed solutions were shown to be asymptotically optimal as the number of subcarriers approaches infinity, based on the frequency-domain virtual time-sharing argument \cite{Yu2006}.
For relay-assisted multi-user scenarios, joint optimization of channel-user assignment and power allocation was considered in \cite{Yu2007} for communication between a base station and users who have the ability to relay information for each other. Based on the same virtual time-sharing argument \cite{Yu2006}, asymptotically optimal solution was provided for network utility maximization.

The problem is especially challenging when channel pairing, channel-user assignment, and power allocation need to be optimized jointly in relay-assisted multi-user scenarios. Existing work to tackle it has been scarce. In \cite{Han2009}, such joint optimization was considered for cooperation among users in uplink communication, accounting for the splitting of bandwidth at a user that needs to simultaneously transmit its own data and relay for others. The proposed problem was NP hard and a suboptimal heuristic algorithm was constructed. The authors of \cite{Zhou2009} studied this problem for a single relay using DF \emph{without} the direct source-destination link. Under a total power constraint, they showed that, for sum-rate maximization, it is optimal to \textit{separately} design channel-user assignment, channel pairing, and power allocation. However, this approach is suboptimal for the general case when the direct link is available, when the user weights are non-uniform, or when individual power constraints are considered.  In comparison, we consider more general relaying strategies that use the direct source-destination link, so that the simple pairing scheme based on sorted channel gain is no longer optimal. Furthermore, our proposed approach accommodates individual power constraints in addition to total power constraints, relaying strategies other than DF, and other optimization objectives.  In Section \ref{sec_simulation}, we further illustrate with numerical data that there is a large performance gap between such a separate optimization approach and the jointly optimal solution.

There are also other studies on resource allocation in multi-channel relaying systems, with different system models from the one presented in this paper (for example, \cite{Awad2008, CuiLau2009, XuInfocom2010, Hua2010, Liu:TWCOM2010}).
Due to the significant complexity in these system models, no general optimal solution has been found. Rather, suboptimal algorithms are proposed with an aim to support satisfactory system performance.  In contrast, in this work we tackle the problem of joint resource optimization in a simpler, single-relay system with multiple users, proposing a provably optimal solution with a formal proof for polynomial-time complexity. Some preliminary results of this study have appeared in \cite{Hajiaghayi-Info11}. This version contains substantial extensions, adding detailed solutions on how to accommodate alternate power constraints and performance objectives, and presenting new derivations, proofs, and numerical results.

\section{System Model and Problem Formulation} \label{systemModel}

We consider the scenario where a source communicates with $K$ users via a single relay as illustrated in Fig. \ref{fig1}.  The available radio spectrum is divided into $N$ equal bandwidth
channels, accessible by all nodes.  We focus on the downlink in our analysis in this paper, but the proposed solution framework can be adopted for the uplink by swapping the roles of the source and the users.

We denote by $h_{i}^{sr}$,  $h_{i}^{rk}$, and $h_{i}^{sk}$ the state of channel $i$, for $1 \leq i \leq N$,  over the first hop between the source and the relay, over the second hop between the relay and user $k$, and over the direct link between the source and user $k$, respectively.  The additive noise on a channel at the relay and user $k$ are modeled as i.i.d.~zero-mean Gaussian random variables with variances $\sigma_{r}^2$ and $\sigma_{k}^2$, respectively.  The channel state is assumed to be available at both the source and the relay, which enables them to dynamically assign channels and allocate power according to channel conditions.

\begin{figure}[t]
\centering
\includegraphics[scale=.5]{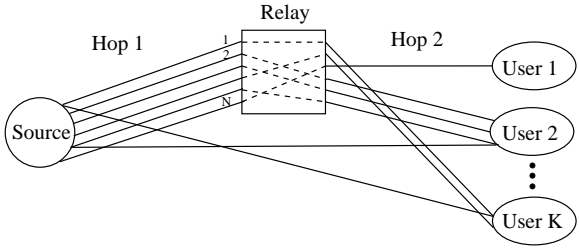}
\vspace*{-1em}
\caption{Illustration of dual-hop multi-channel relaying.} \label{fig1}
\vspace*{-1em}
\end{figure}

\subsection{Channel Assignment}
The relay transmits a processed version of the incoming data to its intended user using a specific relay strategy.  The relay also conducts channel pairing and channel-user assignment. Channel pairing refers to a one-to-one mapping between the incoming channels and outgoing channels at the relay. Through channel-user assignment, on the other hand, a subset of incoming-outgoing channels is assigned to each user. Clearly, channel-pairing choices are closely connected with how the channels are assigned to the users.  We term the joint decision on channel pairing and channel-user assignment the \emph{channel assignment} problem. As the different channels exhibit various quality, judicious channel assignment can potentially lead to significant improvement in spectral efficiency.

We say a path $\mathcal{P}(m,n,k)$ is selected, if first-hop channel $m$ is paired with second-hop channel $n$, and the pair of channels $(m,n)$ is assigned to user $k$.  We define indicator functions $\phi_{mnk}$ for channel assignment as follows:
\beq
\phi_{mnk} = \begin{cases}
1, & \text{if $\mathcal{P}(m,n,k)$ is selected}, \\
0, & \text{otherwise} ~.
\end{cases}
\eeq
There is a one-to-one mapping between first-hop channels and second-hop channels. Furthermore, we require that each channel pair be assigned to only one user, but a user may be assigned multiple channel pairs.  Hence $\phi_{mnk}$ is constrained by
\begin{align}
\sum_{n=1}^{N} \sum_{k=1}^K \phi_{mnk} = 1, \forall m , \quad
\sum_{m=1}^{N} \sum_{k=1}^K  \phi_{mnk} = 1, \forall n . \label{SP_SAconst}
\end{align}

\subsection{Power Allocation}
Along any path $\mathcal{P}(m,n,k)$, the source and relay transmission powers are denoted by $P^s_{mnk}$ and $P^r_{mnk}$, respectively.
We consider both \emph{individual power constraints},
\begin{align}
\sum_{n=1}^{N} \sum_{m=1}^{N} \sum_{k=1}^K  P^s_{mnk} \leq P_s, \quad
\sum_{n=1}^{N} \sum_{m=1}^{N} \sum_{k=1}^K  P^r_{mnk} \leq P_r, \label{eq_ind_power}
\end{align}
and the \emph{total power constraint},
\begin{equation}
\sum_{n=1}^{N} \sum_{m=1}^{N} \sum_{k=1}^K  (P^s_{mnk} + P^r_{mnk}) \leq P_t, \label{eq_totalpower}
\end{equation}
where $P_s$, $P_r$, and $P_t$ are the maximum allowed transmission power by the source, the relay, and the combined source and relay, respectively.  This is a general representation of the power limitations imposed on the system including, e.g., hardware constraint, legal or regulatory requirement, or energy conservation.  Note that this general representation can be easily tailored to also specify systems with only individual power constraints, or only total power constraint, by setting one or more of $P_s$, $P_r$, and $P_t$ to sufficiently large values.

Each constraint above is either inactive (i.e., at optimality it is satisfied with strict inequality) or active (i.e., at optimality it is satisfied with equality).  We consider the case where all active constraints are \textit{strictly} active, i.e., if the problem is modified by changing the power limits by small amounts, at optimality the constraints remain active.  This is without loss of generality, since any constraint that is active but not strictly active can be made inactive, by increasing the power limit by a small amount, without altering the problem solution.

Define $\pmnkvec = (P^s_{mnk}, P^r_{mnk}, P^s_{mnk}+P^r_{mnk})$ and $\Pvec = (P_s,P_r,P_t)$.

\subsection{Relaying Strategy}
We initially focus on DF relaying but will later show how the proposed method can be applied to other relaying schemes, such as AF and CF.
We consider a general case where, apart from the relay path, the direct links are available between the source and users. In this case, the signals received from the relay path and the direct link can be combined to improve the decoding performance.
In DF, each transmission time frame is divided into two equal slots.  In the first slot, the source transmits an information block on each channel, which is received by both the relay and the intended user.  In the second slot, the relay attempts to decode the received message from each incoming channel (first hop), and forwards a version of the decoded message on an outgoing channel (second hop) to the intended user.  The intended user collects the received signals in both time slots, applies maximum ratio combining, and decodes the message.

Consider the conventional repetition-coding based DF relaying \cite{CoverElGamalIT79,laneman04}, where the relay is required to fully decode the incoming message, re-encode it with repetition coding, and forward it to the intended user.  The maximum achievable source-destination rate on path $\mathcal{P}(m,n,k)$ is given by \cite{laneman04}
\begin{align}
R(m,n,k)= \frac{1}{2} &\min \{ \log(1+ a_m P^s_{mnk}), \nonumber\\
& \log(1+ c_{mk} P^s_{mnk} + b_{nk} P^r_{mnk} ) \} ~,
\label{eq_Rmnk}
\end{align}
where $a_m = \frac{|h^{sr}_{m}|^2}{\sigma_r^2}$, $b_{nk} = \frac{|h_{n}^{rk}|^2}{\sigma_k^2}$, and $c_{mk} =\frac{|h_{m}^{sk}|^2}{\sigma_k^2}$ are normalized channel power gains against the noise variance at the relay and user $k$, and the base of logarithm is 2.

\subsection{Optimization Objective}
Various rate-utility functions can be used as objectives. For convenience of illustration, in this paper we focus on the weighted sum-rate.
Denoting by $w_k$ the relative weight for user $k$, such that $\sum_{k=1}^K w_k = 1$, we formulate the problem of weighted sum-rate maximization as
\begin{align}
&\max_{\mathbf{\Phi},\mathbf{P}^s, \mathbf{P}^r} \sum_{k=1}^K  w_k \sum_{n=1}^{N} \sum_{m=1}^{N}  \phi_{mnk} R(m,n,k)  \label{eq_opt_DF} \\
s.t. &  \quad   \eqref{SP_SAconst}, \eqref{eq_ind_power}, \eqref{eq_totalpower}, \nonumber\\
& \quad \phi_{mnk} \in \{0,1\}, \quad \forall m,n,k \label{eq_Phi_bin} \\
& \quad P^s_{mnk} \geq 0, ~P^r_{mnk} \geq 0, \quad \forall m,n,k, \label{eq_P_nonneg}
\end{align}
where $\mathbf{\Phi}\defeq [\phi_{mnk}]_{N\times N\times K}$, $\mathbf{P}^s\defeq [P^s_{mnk}]_{N\times N\times K}$, and $\mathbf{P}^r\defeq [P^r_{mnk}]_{N\times N\times K}$.
Given the relative weights and the channel gains on each path $\mathcal{P}(m,n,k)$, the optimization problem (\ref{eq_opt_DF}) finds the jointly optimal solution of channel pairing, channel-user assignment, and power allocation by optimizing $\mathbf{\Phi}$, $\mathbf{P}^s$, and $\mathbf{P}^r$.

\section{Weighted Sum-Rate Maximization for Multi-channel DF}
\label{sec_totindpower}

The optimization in \eqref{eq_opt_DF} is a mixed-integer programming problem, which in general has intractable complexity due to its combinatorial nature.  However, in this section, we present a method to find an optimal solution with computational complexity growing only polynomially with the number of channels and users.

\subsection{Convex Reformulation via Continuous Relaxation}

The proposed approach is built on the reformulation of \eqref{eq_opt_DF} into a convex optimization problem with a real-valued $\tilde{\mathbf{\Phi}}$ and strong Lagrange duality. We later show that the reformulated problem is optimized by a binary $\mathbf{\Phi}=\tilde{\mathbf{\Phi}}$.

We first substitute
\begin{align}
P^s_{mnk} = \frac{P^s_{mnk}}{\phi_{mnk}}  \quad \text{and} \quad
P^r_{mnk} = \frac{P^r_{mnk}}{\phi_{mnk}}
\label{eq_divphi}
\end{align}
into the objective of \eqref{eq_opt_DF}.  This does not change the original optimization problem, since if $\phi_{mnk}=1$, then \eqref{eq_divphi} is trivially true; and if $\phi_{mnk}=0$, then by l'H\^{o}pital's rule, $\phi_{mnk}R(m,n,k)$ remains zero before and after the substitution.  Indeed, it obviously preserves the optimality of power allocation to enforce $P^s_{mnk} = P^r_{mnk} =0$ for all $(m,n,k)$ such that $\phi_{mnk}=0$.

We then relax the binary constraint on $\mb{\Phi}$ by defining a continuous version of $\phi_{mnk}$, denoted by $\tilde{\phi}_{mnk}$, which may take any value in the interval $[0,1]$.  Then, the reformulated version of the optimization problem \eqref{eq_opt_DF} can be written as
\begin{align}
&\max_{\tilde{\mathbf{\Phi}},\mathbf{P}^s, \mathbf{P}^r}   \sum_{m,n,k}  \frac{w_k}{2} \tilde{\phi}_{mnk}  \min \{ \log(1+ a_m \frac{P^s_{mnk}}{\tilde{\phi}_{mnk}}), \nonumber \\
& \quad \quad \quad \quad \log(1+ c_{mk} \frac{P^s_{mnk}}{\tilde{\phi}_{mnk}} + b_{nk} \frac{P^r_{mnk}}{\tilde{\phi}_{mnk}}) \} \label{OPTP2} \\
s.t.
& \quad \sum_{n,k} \tilde{\phi}_{mnk} = 1, \forall m,
\quad \sum_{m,k} \tilde{\phi}_{mnk} = 1, \forall n, \label{eq_tphi_sum} \\
& \quad 0 \leq \tilde{\phi}_{mnk} \leq 1, \quad \forall m,n,k, \label{eq_tphi_leq1} \\
& \quad \eqref{eq_ind_power}, \eqref{eq_totalpower}, \eqref{eq_P_nonneg}. \nonumber
\end{align}

The objective function \eqref{OPTP2} is concave in $(\tilde{\mathbf{\Phi}},\mathbf{P}^s, \mathbf{P}^r)$, since $\tilde{\phi}_{mnk}\log(1+ a_{m}\frac{P^s_{mnk}}{\tilde{\phi}_{mnk}})$ and $\tilde{\phi}_{mnk}\log(1+ c_{mk} \frac{P^s_{mnk}}{\tilde{\phi}_{mnk}} + b_{nk} \frac{P^r_{mnk}}{\tilde{\phi}_{mnk}})$ are the {\it perspectives}
of the concave functions $\log(1+ a_{m}P^s_{mnk})$ and $\log(1+ c_{mk}P^s_{mnk} + b_{nk}P^r_{mnk})$, respectively%
\footnote{The perspective of function $f: \mb{R}^n \rightarrow \mb{R}$ is defined as $g(x,t) = t f(x/t)$, with domain $\{(x,t) | x/t \in \text{dom} f, t>0 \}$. The perspective operation preserves concavity \cite{Boydbook}. Here we include $\tilde{\phi}_{mnk}=0$ in the domain of the perspectives. It is easy to see that they remain concave.}%
.
It is also noted that the minimum of two concave functions is a concave function. Furthermore, since all the constraints are affine, and there are obvious feasible points, Slater's condition is satisfied \cite{Boydbook}.  Hence, the convex optimization problem \eqref{OPTP2} has zero duality gap, suggesting that a globally optimal solution can be found in the Lagrange dual domain.

Using continuous relaxation on integer programming problems is not a new technique \cite{IP}. However, doing so typically leads only to heuristics or approximations.  Clearly, solving a maximization problem with relaxed constraints generally gives only an upper bound to the original problem. In particular, all global optima for \eqref{OPTP2} do not necessarily give a binary $\tilde{\mathbf{\Phi}}$, which is required for \eqref{eq_opt_DF}.  However, we next show that, in the problem under consideration, indeed there always exists a globally optimal solution to \eqref{OPTP2} consisting of a binary $\tilde{\mathbf{\Phi}}$, and the proposed approach ensures that such an optimal solution is found in polynomial time.

\subsection{Power Allocation  via Maximization of Lagrange Function over $\mb{P}^s$ and $\mb{P}^r$}
\label{sec_powerallocation}

Consider the Lagrange function for (\ref{OPTP2}),
\begin{align}
\mathcal{L}(\tilde{\mathbf{\Phi}},\mb{P}^s,\mb{P}^r,\lvec) = & \sum_{m,n,k} \frac{w_k}{2} \tilde{\phi}_{mnk} \min \{ \log(1+ a_m \frac{P^s_{mnk}}{\tilde{\phi}_{mnk}}), \nonumber \\
& \hspace*{-60pt} \log(1+ c_{mk} \frac{P^s_{mnk}}{\tilde{\phi}_{mnk}} + b_{nk} \frac{P^r_{mnk}}{\tilde{\phi}_{mnk}}) \} - (\l_s+\l_t) \sum_{m,n,k} P^s_{mnk} \nonumber\\
&  \hspace*{-60pt}  - (\l_r+\l_t) \sum_{m,n,k} P^r_{mnk} + \lvec \Pvec^T , \label{LagFun1}
\end{align}
where $\lvec = (\l_s, \l_r, \l_t)$ is the vector of Lagrange multipliers associated with the power constraints \eqref{eq_ind_power} and \eqref{eq_totalpower}. The dual function is therefore
\begin{align}
& g (\lvec) = \max_{\tilde{\mathbf{\Phi}},\mb{P}^s,\mb{P}^r} \mathcal{L}(\tilde{\mathbf{\Phi}},\mb{P}^s,\mb{P}^r,\lvec) \label{dual} \\
s.t. &  \quad   \eqref{eq_tphi_sum}, \eqref{eq_tphi_leq1}, \eqref{eq_P_nonneg} . \nonumber
\end{align}

The above maximization of the Lagrange function can be carried out by first optimizing the power allocation given fixed $\tilde{\mathbf{\Phi}}$.  The KKT conditions suggest that the maximization of \eqref{dual} over $\mb{P}^s$ and $\mb{P}^r$ can be decomposed into $N\times N\times K$ independent subproblems to find the optimal $P^{s*}_{mnk}$ and $P^{r*}_{mnk}$:
\begin{align}
\max_{P^s_{mnk} \geq 0, P^r_{mnk} \geq 0} \mathcal{L}_{mnk}(\tilde{\phi}_{mnk},P^s_{mnk},P^r_{mnk},\lvec)
\label{Subproblem}
\end{align}
where
$\mathcal{L}_{mnk}(\tilde{\phi}_{mnk},P^s_{mnk},P^r_{mnk},\lvec)$
is the part of $\mathcal{L}(\tilde{\mathbf{\Phi}},\mb{P}^s,\mb{P}^r,\lvec)$ that concerns only the path $\mathcal{P}(m,n,k)$.

It can be shown that the solution to \eqref{Subproblem} has the following form.  The derivation details are given in Appendix \ref{appendix_Smnk}.  Note that, since $P^{s*}_{mnk}$ and $P^{r*}_{mnk}$ depends on $\tilde{\phi}_{mnk}$ in an obvious way, we simply present them as functions of $\lvec$ for the rest of this section.
\begin{align}
& \quad \big(P^{s*}_{mnk}(\lvec), P^{r*}_{mnk}(\lvec)\big)
= \nonumber\\
&\begin{cases}
\left( \left[\frac{w_k}{\a(\l_s+\l_t)} - \frac{1}{a_m} \right]^+ \tilde{\phi}_{mnk}, 0\right), \quad \text{if } a_m \leq c_{mk} \\
\mb{p}_1, \quad \text{if } a_m > c_{mk} \text{ and } \frac{c_{mk}}{\l_s+\l_t} < \frac{b_{nk}}{\l_r+\l_t} \\
\displaystyle{\argmax_{(P^s,P^r)\in\{\mb{p}_1, \mb{p}_2\}}} \mathcal{L}_{mnk} (\tilde{\phi}_{mnk},P^s,P^r,\lvec), \quad \text{o.w.}
\end{cases}
\label{eq_Ssrmnkopt}
\end{align}
where $\a = 2\ln 2$, $[x]^+ = \max\{x,0\}$,
$\mb{p}_1 = \left(1, \frac{a_m-c_{mk}}{b_{nk}} \right) \times \left[ \frac{w_k b_{nk}}{\a (b_{nk}(\l_s+\l_t) +(a_m-c_{mk})(\l_r+\l_t))} - \frac{1}{a_m} \right]^+\tilde{\phi}_{mnk}  $, and
$\mb{p}_2 = \left(\left[\frac{w_k}{\a (\l_s+\l_t)} - \frac{1}{c_{mk}}\right]^+ \tilde{\phi}_{mnk} , 0 \right)$.

\subsection{Channel Assignment via Maximization of Lagrange Function over $\tilde{\mathbf{\Phi}}$}
\label{sec_op_Phi}

To maximize the Lagrange function over $\tilde{\mathbf{\Phi}}$, we define
\begin{align}
A_{mnk}(\lvec) = \frac{1}{\tilde{\phi}_{mnk}}\mathcal{L}_{mnk}(\tilde{\phi}_{mnk},P^{s*}_{mnk}(\lvec),P^{r*}_{mnk}(\lvec),\lvec) ~.
\label{Amnk}
\end{align}
Note that $A_{mnk}(\lvec)$ is independent of $\tilde{\phi}_{mnk}$ because of the multiplication form of \eqref{eq_Ssrmnkopt} by $\tilde{\phi}_{mnk}$.
Then, \eqref{dual} can be determined by the following optimization problem over $\tilde{\mathbf{\Phi}}$:
\begin{align}
& \max_{\tilde{\mathbf{\Phi}}} \sum_{m,n,k} \tilde{\phi}_{mnk} A_{mnk}(\lvec) \label{OPTP3} \\
s.t. &  \quad   \eqref{eq_tphi_sum}, \eqref{eq_tphi_leq1} . \nonumber
\end{align}
To proceed, we present the following lemma on the decomposition of $\tilde{\mathbf{\Phi}}$.

\begin{lemma} \label{lemma1}
Any matrix $\tilde{\mathbf{\Phi}}=[\tilde{\phi}_{mnk}]_{N\times N\times K}$ with $0\leq \tilde{\phi}_{mnk} \leq 1$ and satisfying  \eqref{eq_tphi_sum} can be decomposed into one matrix $\mb{X}=[x_{mn}]_{N\times N}$ and $MN$ vectors $\mb{y}^{mn}=[y_k^{mn}]_{1\times K}$, such that $\tilde{\phi}_{mnk}= x_{mn} y_k^{mn}, \forall m,n,k$, with $0 \leq x_{mn} \leq 1$ and $0\leq y_k^{mn} \leq 1$, satisfying
$\sum_{n} x_{mn} =1, \forall m$, $\sum_{m} x_{mn} =1, \forall n$,  and
$\sum_{k} y_k^{mn} =1, \forall m,n$.
Furthermore, any such matrix $\mb{X}$ and vectors $\mb{y}^{mn}$ uniquely determines a matrix $\tilde{\mathbf{\Phi}}$ that is given by $\tilde{\phi}_{mnk}= x_{mn} y_k^{mn}$ and satisfies \eqref{eq_tphi_sum}.
\end{lemma}
\begin{IEEEproof}
The proof is provided in Appendix \ref{appendix1}.
\end{IEEEproof}

Note that, even though the above decomposition can also be applied to a binary $\mathbf{\Phi}$ as a trivial special case of Lemma \ref{lemma1}, we require the general form of this lemma to deal with continuous values in $\tilde{\mathbf{\Phi}}$, $\mb{X}$, and $\mb{y}^{mn}$.  In particular, the mapping from $\tilde{\mathbf{\Phi}}$ to $(\mb{X}, \{\mb{y}^{mn}\})$ is one-to-many, which is quite different from the binary case.

Lemma \ref{lemma1} implies that any optimization over ($\mb{X}$, $\mb{y}^{mn}$) also optimizes $\tilde{\mathbf{\Phi}}$ for the same objective.  This allows us to replace, in problem (\ref{OPTP3}), $\tilde{\phi}_{mnk}$ with  $x_{mn}y_k^{mn}$. Furthermore, the constant terms can be dropped from (\ref{OPTP3}).  Hence, we can equivalently seek solutions to the following problem
\begin{align}
& \max_{\mb{X},\{\mb{y}^{mn}\}}  \sum_{m,n} x_{mn} \sum_k y_k^{mn} A_{mnk}(\lvec)  \label{OPT_Amnk} \\
s.t.
& \quad \sum_n x_{mn} =1, \forall m, \sum_m x_{mn} =1, \forall n,   \quad 0 \leq x_{mn} \leq 1, \forall m,n, \label{eq_x_const}\\
& \quad \sum_k y_k^{mn} = 1, \forall m,n,  \quad 0 \leq y_k^{mn} \leq 1, \forall m,n,k . \label{eq_y_const}
\end{align}
The following two-stage solution is sufficient.
First, the inner-sum term is maximized over $y_k^{mn}$ for each $(m,n)$ pair, i.e.,
\begin{align}
& A'_{mn}(\lvec) = \max_{\mb{y}^{mn}} \sum_k y_k^{mn} A_{mnk}(\lvec)  \label{OPT_Yk} \\
s.t. &  \quad \eqref{eq_y_const} . \nn
\end{align}
An optimal solution to \eqref{OPT_Yk} is readily obtained as
\beq
{y_k^{mn}}^* =
\begin{cases}  1, & \quad \text{if } k = \argmax_{1\leq l \leq K} A_{mnl}(\lvec) \\
 0,  & \quad \text{otherwise}
\end{cases}~.
\label{Opt_Ymnk}
\eeq
In the above maximization, arbitrary tie-breaking can be performed if necessary.
Next, inserting $A'_{mn}(\lvec)$ into (\ref{OPT_Amnk}), we have the linear optimization problem
\begin{align}
&\max_{\mb{X}}  \sum_{m,n} x_{mn} A'_{mn}(\lvec)  \label{OPTXmn} \\
s.t. & \quad \eqref{eq_x_const} \nn.
\end{align}
It is well known that there always exists an optimal solution to \eqref{OPTXmn} that is binary \cite[Chapter~3]{IP}. An intuitive explanation is the following. Since \eqref{OPTXmn} is a linear program with a bounded objective, an optimal solution can be found at the vertices of the feasible region.  Furthermore, since $\mb{X}$ is a doubly stochastic matrix, it is a convex combination of permutation matrices.  One of these vertex permutation matrices is an optimal solution to \eqref{OPTXmn}, so at least one optimal $\mb{X}$ is binary.  Then, to find a binary optimal $\mb{X}$, \eqref{OPTXmn} is a \emph{two-dimensional assignment problem}.  Efficient algorithms, such as the Hungarian Algorithm \cite{Hungarian}, exist to produce an optimal solution with computational complexity being polynomial in $N$.

Finally, the optimal $\tilde{\phi}_{mnk}$ given $\lvec$ is
\beq
\tilde{\phi}^*_{mnk}(\lvec) = x^*_{mn}(\lvec){y_k^{mn}}^*(\lvec) ~.  \label{eq_phi_opt}
\eeq
Since binary $x^*_{mn}(\lvec)$ and ${y_k^{mn*}}(\lvec)$ are computed following the above procedure, $\tilde{\phi}^*_{mnk}(\lvec)$ is also binary. This shows that there exists at least one binary optimal solution to the maximization in \eqref{OPTP3}.


%
Intuitively, the globally optimal solution described above suggests a pairing between the input and output channels at the relay, and if channels $m$ and $n$ are paired, they are assigned to a single user $k$, whose associated $A_{mnk}(\lvec)$ is the greatest among all users.  Note that such an interpretation might lead us to conclude that we could have forgone continuous relaxation from the very beginning and focused only on a binary $\mb{\Phi}$. However, we would still have required the continuous $\tilde{\mathbf{\Phi}}$ to construct a convex optimization problem, whose strong duality property provides the optimality of the proposed approach.  The optimality of $(\mb{X}, \{\mb{y}^{mn}\})$ taking binary values is implied only through the above derivation.

Interestingly, the original optimization problem (\ref{OPTP3}) with a binary matrix $\mb{\Phi}$ is a special case of the \emph{axillary three-dimensional assignment problems} \cite{MultiAP}.  It is well known that the general form of this family of problems is NP hard and cannot be solved by continuous relaxation on $\mb{\Phi}$, unlike the two-dimensional assignment problem in \eqref{OPTXmn}.  In our case, the special structure of $\tilde{\mathbf{\Phi}}$ expressed in \eqref{eq_tphi_sum}, namely the absence of a constraint on per-user resource allocation, makes possible the availability of an efficient solution to \eqref{OPTP3}.

It is also worth noting that, given any $\lvec$, there may exist non-integer optimal solutions to \eqref{OPTP3}.  For example, when the maximal value of $A_{mnk}$ in (\ref{OPT_Yk}) is achieved by multiple users having the same channel gains, there is an infinite number of optimal $\mb{y}^{mn}$, leading to non-integer optimal solutions for $\tilde{\phi}_{mnk}$.  However, the procedure above finds only one of the optimal solutions in binary form, which is sufficient for computing the dual function.


\subsection{Dual Minimization: Baseline Subgradient Approach}

The previous subsection provides a way to find the Lagrange dual $g(\lvec)$ for any Lagrange multiplier vector $\lvec$.  Next, the standard approach calls for minimizing the dual function:
\begin{align}\label{dual2}
&\min_{\lvec} g(\lvec)  \\
s.t. &\quad \lvec \succeq  \mb{0}. \nonumber
\end{align}
This can be solved using the subgradient method \cite{shor85}.
It is easy to verify that a subgradient at the point $\lvec$ is given by
\beq
\thetavec(\lvec) = \Pvec - \sum_{m,n,k} \pmnkvec^*(\lvec)  ~, \label{subgradient}
\eeq
where $P^{s*}_{mnk}(\lvec)$ and $P^{r*}_{mnk}(\lvec)$ are computed based on \eqref{eq_Ssrmnkopt} and  $\tilde{\phi}^*_{mnk}(\lvec)$ found using \eqref{eq_phi_opt}.

For completeness, we first summarize the standard subgradient updating algorithm for solving the dual problem in the following. We will present a modified dual minimization algorithm in Section \ref{sec_opt_comp}, which is guaranteed to converge in polynomial time.

\begin{enumerate}
  \item Initialize $\lvec^{(0)}$.
 \item Given $\lvec^{(l)}$,  obtain the optimal values of $P^{s*}_{mnk}(\lvec^{(l)})$, $P^{r*}_{mnk}(\lvec^{(l)})$, and $\tilde{\phi}_{mnk}^*(\lvec^{(l)})$.
 \item Update $\lvec$ through
$\lvec^{(l+1)} = [\lvec^{(l)} - \thetavec(\lvec^{(l)}) \nu^{(l)}]^+$
where $\nu^{(l)}$ is the step size at the $l$th iteration.
\item Let $l=l+1$; repeat from Step 2) until the convergence of $\min_l g(\lvec^{(l)})$.
\end{enumerate}

Several step-size rules have been proven to guarantee convergence under some general conditions \cite{shor85}\cite{SubgradientBoyd}. For example, using a constant step size $\nu$, i.e., $\nu^{(l)}=\nu$, or using a constant step length $\nu$, i.e., $\nu^{(l)}=\nu / \|\thetavec(\lvec^{(l)})\|_2$, leads to an objective within a given neighborhood of a global optimum; while using the non-summable, square-summable rule leads to asymptotic convergence to a global optimum.  Furthermore, one may satisfy any constraints on $\lvec^{(l)}$ within a convex region by projecting $\lvec^{(l)}$ onto the region.  This is the general \textit{projected} subgradient method, which does not reduce the speed of convergence \cite{SubgradientBoyd}.  For example, Step 3 above ensures that $\lvec^{(l)} \succeq \mb{0}$, and we will further consider projection onto convex regions $\mathcal{R}_1$ and $\mathcal{R}_2$ in Section \ref{sec_opt_comp}.

\subsection{Primal Optimality}

With standard subgradient updating, the dual optimal $\lvec^*$ is obtained, from which we compute the channel assignment and power allocation matrices $(\mathbf{\Phi}^*, \mathbf{P}^{s*}, \mathbf{P}^{r*})$, where $\mathbf{\Phi}^*=\tilde{\mathbf{\Phi}}^*$.  Since the optimization problem \eqref{OPTP2} is a convex program that satisfies Slater's condition, it has zero duality gap.  Denote by $f^*(\Pvec)$ the maximal value of the objective in \eqref{OPTP2}.  Then $f^*(\Pvec)=g(\lvec^*)$, and furthermore it is concave in $\Pvec$.  We consider systems that have the following \textit{strictly} diminishing rate-power relation:
\begin{assumption} \label{assumption}
$f^*(\Pvec)$ is strictly concave in any strictly active power constraint $P_x \in \{P_s,P_r,P_t\}$.
\end{assumption}

In other words, as the
data rate increases, each unit of increment requires more and
more marginal power.  With a strictly concave $R(m,n,k)$ in terms of $\pmnkvec$, this assumption holds when either there is no tie-breaking in \eqref{Opt_Ymnk} or \eqref{OPTXmn} or there is tie-breaking that is due to users or paths having the same weights or channel gains%
\footnote{However, we cannot rule out the possibility of a case where other forms of tie-breaking in \eqref{Opt_Ymnk} or \eqref{OPTXmn} might create linear segments in $f^*(\Pvec)$, although in all simulation tests with arbitrary parameters, we have not produced a case where this assumption fails.}%
.

\begin{proposition} \label{prop_opt_std}
Under Assumption \ref{assumption}, $(\mathbf{\Phi}^*, \mathbf{P}^{s*}, \mathbf{P}^{r*})$ is a globally optimal solution to the original problem \eqref{eq_opt_DF}.
\end{proposition}
\begin{proof}
Since $\mathbf{P}^{s*}$ and $\mathbf{P}^{s*}$ are uniquely determined by $\lvec^*$ and $\mathbf{\Phi}^*$, we need only to focus on $\mathbf{\Phi}^*$.

For any inactive constraint $P_x$, we have $\l_x^*=0$ and the subgradients of $g(\lvec)$ in the direction of $\l_x$ are all positive. Hence any $\mathbf{\Phi}^*$ is feasible with respect to $P_x$.

For any strictly active constraint $P_x$, we have  $\l_x^*>0$. Furthermore,
\begin{align}
&f^*(\Pvec)
=  \mathcal{L}(\mathbf{\Phi}^*, \mathbf{P}^{s*}, \mathbf{P}^{r*}, \lvec^*) \nn\\
\leq & f^*(\sum_{m,n,k} \pmnkvec^*) - \lvec^*(\sum_{m,n,k} \pmnkvec^* - \Pvec)^T ~.
\end{align}
Given Assumption \ref{assumption}, the above is possible only when all strictly active constraints are satisfied with equality. Therefore, any $\mathbf{\Phi}^*$ is feasible with respect to $P_x$ and the complementary slackness condition is satisfied. Hence, $(\mathbf{\Phi}^*, \mathbf{P}^{s*}, \mathbf{P}^{r*})$ is a globally optimal solution to \eqref{OPTP2}.

Furthermore, since \eqref{OPTP2} is a constraint-relaxed version of \eqref{eq_opt_DF}, $\mathbf{\Phi}^*$ gives an upper bound to the objective of \eqref{eq_opt_DF}.  Finally, since $\mathbf{\Phi}^*$ satisfies the binary constraints in \eqref{eq_opt_DF} at each iteration of the subgradient algorithm, it satisfies all constraints in \eqref{eq_opt_DF}.  Therefore, it is a globally optimal solution to \eqref{eq_opt_DF}.
\end{proof}

We point out that using conventional convex optimization software packages directly on the relaxed problem \eqref{OPTP2} is not sufficient to solve \eqref{eq_opt_DF}.  This is because there is no guarantee that they will return a binary $\tilde{\mathbf{\Phi}}^*$, and furthermore due to complicated three-dimensional dependencies among $\phi_{mnk}$, there is no readily available method to transform a fractional $\tilde{\mathbf{\Phi}}^*$ to the desired binary solution.

\subsection{Dual Minimization: Divide-and-Conquer Algorithm with Polynomial Complexity} \label{sec_opt_comp}

The standard subgradient method produces a global optimum, but its computational complexity is not generally known.  Previous studies have provided asymptotic bounds or  conjectures on its efficiency through computational experience. In general, the number of iterations in subgradient updating depends on the step-size rule, the distance between the initial solution and the optimal solution, and the 2-norm of the subgradients \cite{shor85,SubgradientBoyd}.

Next, we propose a new dual minimization algorithm that guarantees convergence with polynomial complexity in $N$ and $K$, to a global optimum for our optimization problem.  It uses a divide-and-conquer approach, by grouping the possible locations of $\lvec^*$ into two regions and applying projected subgradient updating constrained within either.  It ensures that in each region, our choice of the initial $\lvec^{(0)}$ and subsequent subgradient updating lead to convergence in polynomial time.

We first define the following two overlapping convex regions in terms of $\lvec$:
\begin{align}
\mathcal{R}_1 &\defeq \bigg\{\lvec: \l_s + \l_t \geq \nn\\
&\frac{\displaystyle \min_{\{k: w_k>0\}} w_k  \min \{ \min_{\{m: a_m>0\}} a_m, \min_{\{m,k: c_{mk}>0\}} c_{mk} \}}{\displaystyle 4\a (\max_m a_m \min\{P_s,P_t\} + 1)}, \lvec \succeq \mb{0} \bigg\} ~, \nn\\
\mathcal{R}_2 &\defeq \bigg\{\lvec: \l_s + \l_t + \frac{\displaystyle \min_{\{m: a_m>0\}} a_m }{\displaystyle \max_{n,k} b_{nk}} (\l_r + \l_t) \geq \nn\\
&\quad \quad \frac{\displaystyle \min_{\{k: w_k>0\}} w_k}{\displaystyle \a \big(\min\{P_s,P_t\} + \frac{1}{\displaystyle \min_{\{m: a_m>0\}} a_m}\big) }, \lvec \succeq \mb{0}  \bigg\} ~. \nn
\end{align}
These are two possible regions where $\lvec^*$ resides, which depends on whether there exists at least one \textit{chosen} path with non-zero direct-link channel gain $c_{mk}$. This is formalized in the following lemma.  Its proof is given in \cite{HajiaghayiDongLiang:JSACarXiv2011}.
\begin{lemma} \label{lem_lambda_regions}
If there exists some $(m,n,k)$ such that $\phi^*_{mnk} = 1$ and $c_{mk}>0$, then $\lvec^* \in \mathcal{R}_1$. Otherwise, $\lvec^* \in \mathcal{R}_2$.
\end{lemma}

The proposed divide-and-conquer dual minimization (DCDM) algorithm considers both possible regions for $\lvec^*$.  It first creates the two conditions in Lemma \ref{lem_lambda_regions} by artificially setting direct-link channel gains to zero.  It then applies the projected subgradient algorithm on $\mathcal{R}_1$ and $\mathcal{R}_2$ separately, and chooses the better solution between these two.  The algorithm is formally detailed in Algorithm \ref{alg_subgradient}, and its optimality and complexity are given in Propositions \ref{prop_opt_mod} and \ref{thm:complexity}, respectively.  Note that one cannot use Lemma \ref{lem_lambda_regions} to determine, \textit{before} the optimal channel assignment matrix $\tilde{\mathbf{\Phi}}$ is chosen, which region $\lvec^*$ is in. This necessitates the comparison step in the DCDM algorithm.

\begin{algorithm}[tbp]
\begin{algorithmic}
\IF {there exists some $m$ and $k$ such that $c_{mk}>0$}
        \STATE $\lvec^*_1 =$ output of subgradient updating algorithm with projection onto $\lvec^{(l)} \in \mathcal{R}_1$
        \STATE Set $c_{mk}=0$ for all $1 \leq m \leq N$ and $1 \leq k \leq N$
        \STATE $\lvec^*_2 =$ output of subgradient updating algorithm with projection onto $\lvec^{(l)} \in \mathcal{R}_2$
        \RETURN $\displaystyle \argmin_{\lvec \in \{ \lvec^*_1, \lvec^*_2 \}}  g(\lvec) $
\ELSE
        \STATE $\lvec^* =$ output of subgradient updating algorithm with projection onto $\lvec^{(l)} \in \mathcal{R}_2$
        \RETURN $\lvec^*$
\ENDIF
\end{algorithmic}
\caption{Divide-and-Conquer Dual Minimization (DCDM)}
\label{alg_subgradient}
\end{algorithm}

\begin{proposition} \label{prop_opt_mod}
With DCDM, the computed channel assignment and power allocation matrices $(\mathbf{\Phi}^*, \mathbf{P}^{s*}, \mathbf{P}^{r*})$, where $\mathbf{\Phi}^*=\tilde{\mathbf{\Phi}}^*$, is a globally optimal solution to the original problem \eqref{eq_opt_DF}.
\end{proposition}
\begin{proof}
Suppose there exists some $(m,n,k)$ such that $\phi^*_{mnk} = 1$ and $c_{mk}>0$. Then Lemma \ref{lem_lambda_regions} shows that $\lvec^* \in \mathcal{R}_1$. Therefore, by Proposition \ref{prop_opt_std}, $\lvec^*_1$ obtained by subgradient updating projected onto $\mathcal{R}_1$ is an optimal solution. Furthermore, setting $c_{mk}=0$ for all $1 \leq m \leq N$ and $1 \leq k \leq N$ only reduces $R(m,n,k)$ for all paths, so that subsequently minimizing the Lagrange dual yields an inferior solution. Therefore, $\argmin_{\lvec \in \{ \lvec^*_1, \lvec^*_2 \}}  g(\lvec) = \lvec^*_1$ is returned by DCDM.

Suppose $c_{mk}=0$ for all $(m,n,k)$ such that $\phi^*_{mnk} = 1$, i.e., all chosen paths have zero direct-link channel gain. Then, setting $c_{mk}=0$ for all $1 \leq m \leq N$ and $1 \leq k \leq N$ only reduces $R(m,n,k)$ for the non-chosen paths.  Subsequently minimizing the Lagrange dual yields the same solution as before changing $c_{mk}$. Furthermore, Lemma \ref{lem_lambda_regions} shows that this optimal solution is in $\mathcal{R}_2$.  Hence, $\lvec^*_2$ obtained by subgradient updating projected onto $\mathcal{R}_2$ is an optimal solution. In this case, $\argmin_{\lvec \in \{ \lvec^*_1, \lvec^*_2 \}}  g(\lvec) = \lvec^*_2$ is returned by DCDM.
\end{proof}

The polynomial computational complexity of DCDM is stated in Proposition \ref{thm:complexity}.  Its proof requires the following lemmas, which give upper bounds on $\|\lvec^*\|_2$ and $\|\thetavec(\lvec^{(l)})\|_2$, where $\|\cdot\|_2$ denotes the 2-norm.
\begin{lemma}  \label{lem_ub_lambda}
At global optimum,
$\|\lvec^*\|_2$ is upper bounded by $\l_{max} = O(N^2)$.
\end{lemma}
\begin{IEEEproof}
The proof is provided in Appendix \ref{proof_ub_lambda}.
\end{IEEEproof}

\begin{lemma} \label{lem_ub_theta}
At every step of subgradient updating in the DCDM algorithm, $\|\thetavec(\lvec^{(l)})\|_2$ is upper bounded by $\theta_{max} = O(N^2)$.
\end{lemma}
\begin{IEEEproof}
The proof is provided in Appendix \ref{proof_ub_theta}.
\end{IEEEproof}

\begin{proposition}  \label{thm:complexity}
To achieve a weighted sum-rate within an arbitrary $\epsilon >0$ neighborhood of the optimum $g(\lvec^*)$, using either a constant step size or a constant step length in subgradient updating, the DCDM algorithm has polynomial computational complexity in $N$ and $K$.
\end{proposition}
\begin{proof}
At each iteration of the standard subgradient updating algorithm, the procedures described in Sections \ref{sec_powerallocation} and \ref{sec_op_Phi} are employed. This has computational complexity polynomial in $N$ and $K$.  Therefore, it remains to show that the total number of iterations is not more than polynomial in $N$ or $K$.

For either case of projecting onto $\mathcal{R}_1$ or $\mathcal{R}_2$,
one may choose an initial $\lvec^{(0)}$ such that the distance between $\lvec^{(0)}$ and $\lvec^*$ is upper bounded by $\l_{max}$.
Then, it can be shown that, at the $l$th iteration, the distance between the current best objective to the optimum objective $g(\lvec^*)$ is upper bounded, by
$\frac{\l_{max}^2 + \nu^2 \theta_{max}^2  l}{2 \nu l}$
if a constant step size is used (i.e., $\nu^{(l)}=\nu$), or by
$\frac{\l_{max}^2 \theta_{max}  + \nu^2 \theta_{max} l}{2 \nu l}$ if a constant step length is used (i.e., $\nu^{(l)}=\nu / \|\thetavec(\lvec^{(l)})\|_2$)  \cite{shor85}\cite{SubgradientBoyd}.  For the former and latter bounds, if we set $\nu = \epsilon / \theta_{max}^2$ and $\nu = \epsilon / \theta_{max}$ respectively, both are upper bounded by $\epsilon$  when $l \geq \l_{max}^2 \theta_{max}^2 / \epsilon^2 = O(N^4)$. Hence, the number of required iterations until convergence, for either of the two projected subgradient updating procedures in DCDM, is polynomial in $N$ and independent of $K$.
\end{proof}

Note that using the non-summable, square-summable step-size rule in the early iterations of subgradient updating, often leads to faster movement toward a global optimum than using a constant step size or a constant step length. This is due to its larger step sizes when $l$ is small.  However, such a step-size rule does not guarantee polynomial convergence time%
\footnote{Consider the following idealized example for illustration. If $\nu^{(l)}=\frac{1}{l}$ for all $l$, the number of iterations would need to be $L = \Theta(e^{\l_{max}})$ to satisfy the convergence requirement $\sum_{l=1}^{L} \nu^{(l)} = \Theta(\l_{max})$.}%
.  Therefore, one may start with the non-summable, square-summable rule, and then switch to one of the constant-step rules when the step size or step length is sufficiently near the prescribed value in Proposition \ref{thm:complexity}.  This would reduce the convergence time in practice while preserving the guarantee of polynomial complexity.

\section{Extensions to General Relaying Strategies}
\label{sec_extensions}

For any relaying strategy in which data sent through different communication paths $\mathcal{P}(m,n,k)$ are independent and the achievable rates $R(m,n,k)$ is a concave function in transmission powers $(P^s_{mnk}, P^r_{mnk})$, the proposed solution approach gives jointly optimal channel assignment and power allocation for weighted sum-rate maximization.  To see this, we first note that any concave rate function would lead to convex programming for the relaxed and reformulated problem, which satisfies Slater's condition and hence has zero duality gap. Furthermore, toward maximizing the Lagrange function, we can generalize \eqref{Subproblem} into the following form:
\begin{align}
\max_{P^s_{mnk}\geq 0,P^r_{mnk}\geq 0}  \quad w_k \tilde{\phi}_{mnk} R(\frac{P^s_{mnk}}{\tilde{\phi}_{mnk}}, \frac{P^r_{mnk}}{\tilde{\phi}_{mnk}})  - \lvec \pmnkvec^T ~. \nn
\end{align}
Since the partial derivatives of the above maximization objective contains $P^s_{mnk}$ and $P^r_{mnk}$ only in the form of $\frac{P^s_{mnk}}{\tilde{\phi}_{mnk}}$ and $\frac{P^r_{mnk}}{\tilde{\phi}_{mnk}}$, we always have $P^{s*}_{mnk}$ and $P^{s*}_{mnk}$ as the product of $\tilde{\phi}_{mnk}$ and a non-negative factor.
This leads to a maximization problem of the form in \eqref{OPTP3}, which has been shown to admit a binary optimal solution in Section \ref{sec_op_Phi}.

Besides DF, the time-sharing variants of any relaying strategies with long-term or short-term average power constraints, as well as all capacity achieving strategies, have concave achievable rates \cite{ElGamal06TransIT}.  Our algorithm is applicable to these current and future relaying strategies to find the optimal solution.  However, the closed-form solutions for $(P^{s*}_{mnk},P^{r*}_{mnk})$ may be difficult to find in some cases, requiring more involved numerical computation.

For relaying strategies that do not have concave achievable rates, such as AF, near-optimal solutions can be obtained by using the proposed approach in the following senses:
\begin{list}{$\bullet$}{ \setlength{\itemsep}{0pt}
     \setlength{\parsep}{0pt}
     \setlength{\topsep}{0pt}
     \setlength{\partopsep}{0pt}
     \setlength{\leftmargin}{1.5em}
     \setlength{\labelwidth}{1em}
     \setlength{\labelsep}{0.5em} }
\item A concave bound of the achievable rate may be used to approximate $R(m,n,k)$. For example, with AF, we have $
R(m,n,k) = \frac{1}{2} \log(1+ \frac{a_m b_{nk}P^r_{mnk} P^s_{mnk} }{1+ a_mP^s_{mnk} + b_{nk}P^r_{mnk} } + P^s_{mnk} c_{mk} ) $. A concave upperbound is obtained by removing ``1'' from the denominator.  By substituting such a concave bound for $R(m,n,k)$ in the original optimization problem, we obtain a solution that optimizes in terms of the bound. In the case of AF, such solution is near-optimal for weighted sum-rate, since the ``1'' is negligible for paths with high effective SNR, while paths with low effective SNR do not contribute substantially to the performance objective.
\item It has been shown in \cite{Yu2006} that, regardless of the convexity of the objective function in a multi-channel resource assignment problem, if the objective \emph{at optimum} is a concave function of the maximum allowed powers, the duality gap of the Lagrange dual induced by power constraints is zero. This is due to time-sharing over resource assignment strategies.  Furthermore, there is a frequency-domain approximation of time-sharing, so that the duality gap is asymptotically zero when the number of channels goes to infinity.  Hence, for systems with a large number of channels, near-optimal results can be achieved by the proposed approach.
\end{list}

Finally, if we consider $R(m,n,k)$ as a general concave utility function of the rate on path $\mathcal{P}(m,n,k)$, then for concave and increasing rates, the utility function is also concave in the optimization variables, so that a similar optimization approach is applicable. An example is the weighted $\alpha$-fair utility function \cite{mo00}, which represents general fairness targets such as proportional fairness and max-min fairness. Note that this provides only fairness among the paths, instead of among users who could be assigned multiple paths.

\section{Numerical Results}
\label{sec_simulation}

In this section, we compare the performance of jointly optimal channel pairing, channel-user assignment, and power allocation  with that of suboptimal schemes. We further study the different factors that affect the performance gap under these schemes, in order to shed light on the tradeoff between performance optimality and implementation complexity. The suboptimal schemes considered are
\begin{itemize}
\item {\it No Pairing}: Channel-user assignment and power allocation are jointly optimized, but no channel pairing is performed, i.e.~the same incoming and outgoing channels are assumed. The solution is found by always assigning an identity matrix to $\mb{X}$ instead of solving \eqref{OPTXmn}.
\item {\it No PA}: Allocate power uniformly across all channels, subject to power constraints. Channel pairing and channel-user assignment are jointly optimized by solving \eqref{eq_opt_DF} with given power, which is a three-dimensional assignment problem over binary $\mb{\Phi}$. The solution is found by following the procedure in Section \ref{sec_op_Phi}.
\item {\it Separate Optimization}: The three-stage solution proposed in \cite{Zhou2009}, with channel-user assignment based on maximum channel gain over the second-hop channels, channel pairing based on sorted channel gains, and water-filling power allocation.
\item {\it Max Channel Gain}: Channel-user assignment by maximum channel gain over the second hop, with uniform power allocation and no channel pairing.
\end{itemize}

We use OFDMA as an example for a multi-channel system. The relaying network setup is shown in Fig.~\ref{config}, where the distance between the source and the relay is denoted by $d_{sr}$, and $K=4$ users are located on a half-circle arc around the relay with radius $d_{rd}$.
A $4$-tap frequency-selective propagation channel is assumed for each hop, and the number of channels is set to $N=16$. We define a nominal SNR, denoted by $\SNR_\nom$, as the average received SNR over each subcarrier under uniform power allocation. Specifically, with total power constraint $P_t$, we have  $\SNR_\nom \defeq \frac{{P_t}(\bar{d}_{sd})^{-\kappa}}{2\sigma^2N}$, where  $\kappa=3$ denotes the pathloss exponent, $\sigma^2$ denotes the noise power per channel, and $\bar{d}_{sd}$  denotes the average distance between the source and users. A total power constraint and equal individual power constraints on both the source and the relay are assumed with $P_s=P_r=\frac{2}{3}P_t$, unless it is stated otherwise.

\begin{figure}[tb]
\centering
\includegraphics[scale=.8]{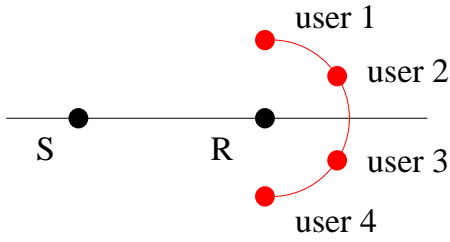}
\caption{Simulation configuration with $K=4$ users}
\label{config}
\end{figure}

\subsection{Performance versus Nominal SNR}

We compare the performance of various channel assignment and power allocation schemes at different $\SNR_\nom$ levels for $K=4$. We fix the ratio $d_{sr}/d_{rd}$ to be $1/3$. Fig.~\ref{EqualWeight} depicts the normalized weighted sum-rate (normalized over $N$) vs.~$\SNR_\nom$ for DF relaying with equal weight, \ie $\mb{w}\defeq [w_k]_{1\times K}=[.25,.25,.25,.25]$. The jointly optimal scheme outperforms the other suboptimal schemes, and provides as much as 20$\%$ gain over the \emph{Separate Optimization} scheme. The gain is increased when an unequal weight vector $\mb{w}$ is required to satisfy different user QoS demands or fairness. Fig.~\ref{NonEqualweight} shows the normalized weighted sum-rate vs.~$\SNR_\nom$ for $\mb{w}= [.15,.15,.35,.35]$, where a substantial gain is observed by employing the jointly optimal solution.

\begin{figure}[tb]
\centering
\includegraphics[scale=.5]{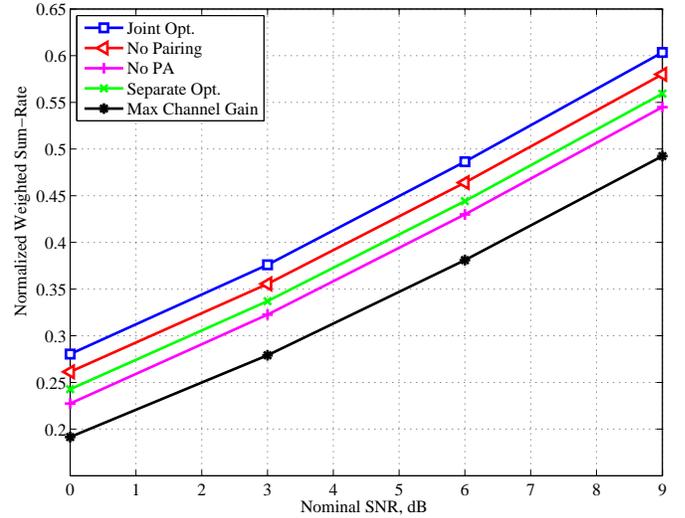}
\caption{Normalized weighted sum-rate vs.~nominal SNR with $w=[.25,.25,.25,.25]$, $N=16$, $K=4$, and DF relaying.} \label{EqualWeight}
\end{figure}

\begin{figure}[tb]
\centering
\includegraphics[scale=.5]{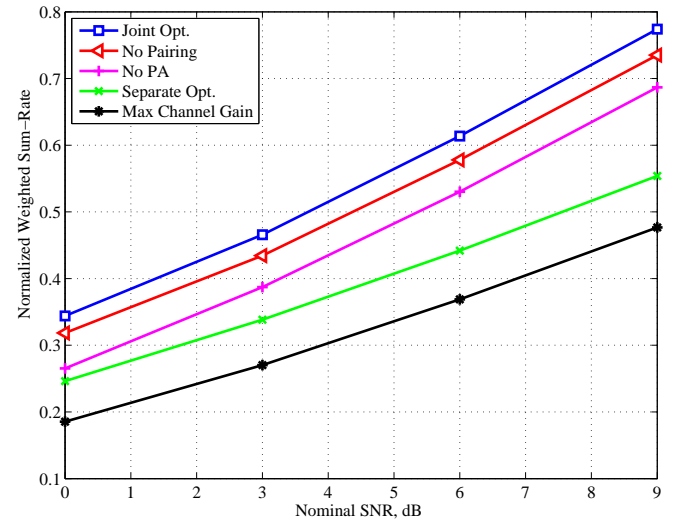}
\caption{Normalized weighted sum-rate vs.~nominal SNR with $w=[.15,.15,.35,.35]$, $N=16$, $K=4$, and DF relaying.}  \label{NonEqualweight}
\end{figure}

\subsection{Performance versus Number of Users}
In this experiment, we show how the number of users affects the performance of various resource assginment schemes. We increase the number of users in Fig~\ref{config}, and uniformly place them around the half-circle arc. In order to properly compare weighted sum-rate under different number of users, we do not normalize $\sum_{k=1}^Kw_k=1$. Instead, we fix $w_k=1$ for all $k$.  The nominal SNR is $\SNR_{\nom}=4$dB, and the ratio $d_{sr}/d_{rd} = 1/3$. Fig~\ref{rate_vs_users} shows the normalized weighted-sum rate vs. the number of users for DF relaying under total power constraint $P_t$. As we see, the sum-rate is improved due to the multi-user diversity gain with an increased number of users. In addition, consistent performance gain under joint optimization can be seen over different user population sizes.

\begin{figure}
\centering
\includegraphics[scale=.5]{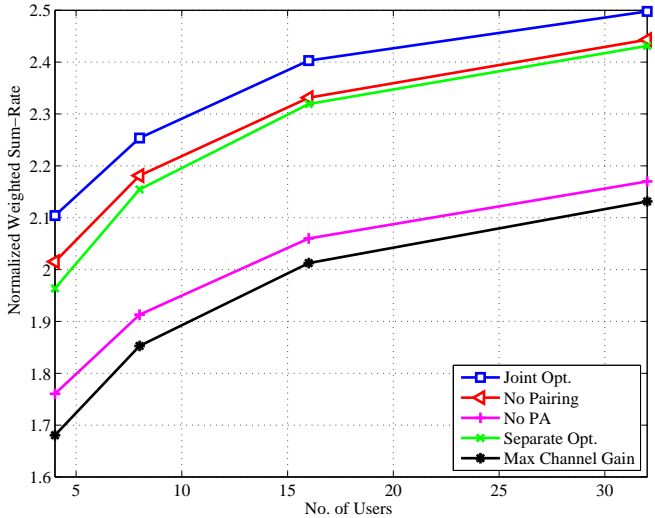}
\caption{Normalized weighted sum-rate vs. number of users with equal weight $w_i=1$, for $1\le i\le K$, $N=16$, and DF relaying.}
\label{rate_vs_users}
\end{figure}

\subsection{Impact of Relay Position}
Through this experiment, we study how the relay position affects the performance under various resource assignment schemes. The $K=4$ users are located close to each other as a cluster, and they have approximately the same distance to the relay and the source. We change the relay position along the path between the source and the user cluster.
Figs~\ref{rate_vs_dr_NeqW} and \ref{rate_vs_dr_NeqW_AF} demonstrate the normalized weighted-sum rate vs.~the ratio $d_{sr}/d_{sd}$.  We set $\mb{w}=[.15,.15,.35,.35]$, and $\SNR_\nom=3$dB. Fig.~\ref{rate_vs_dr_NeqW} shows the DF relaying case under both total and individual power constraints, and Fig.~\ref{rate_vs_dr_NeqW_AF} shows the AF relaying case under a total power constraint.

We see from Fig.~\ref{rate_vs_dr_NeqW} that better performance is observed when the relay is closer to the source than to the users in DF relaying, as correctly decoding data at relay is important in successful DF relaying. In addition, comparing the joint optimal scheme with \emph{No Pairing} scheme, we see that the gain of channel pairing is evident when the relay is closer to the source, but diminishes when the relay moves closer to the users. In the latter case, as the first-hop becomes the bottleneck, channel pairing at the second-hop provides no benefit. This is not the case for AF relaying. As shown in Fig.~\ref{rate_vs_dr_NeqW_AF}, channel pairing gain is observed throughout different relay positions. Furthermore, the performance of the jointly optimal solution only has mild variation throughout different relay positions, unlike the \emph{No PA} scheme. This suggests that the benefit of optimal power allocation for AF relaying is more significant when the relay is closer to either the source or the users.

\begin{figure}[tbph]
\centering
\includegraphics[scale=.5]{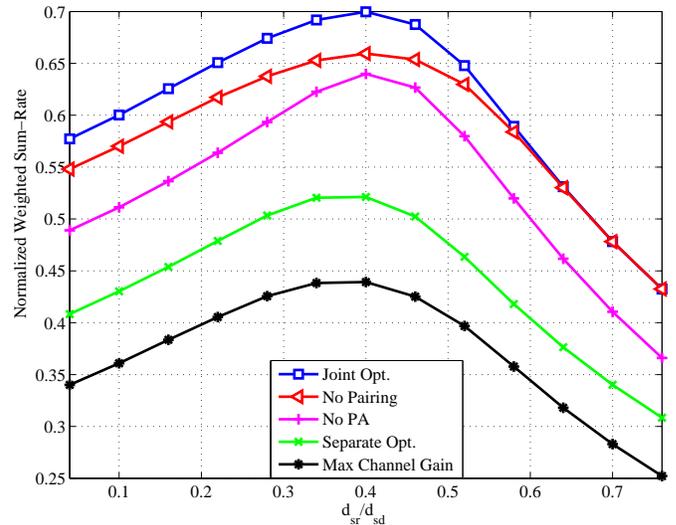}
\caption{Normalized weighted sum-rate vs. relay location; $K=4$, $\mb{w}=[.15,.15,.35,.35]$, $N=16$, and DF relaying.}  \label{rate_vs_dr_NeqW}
\end{figure}

\begin{figure}[tbph]
\centering
\includegraphics[scale=.5]{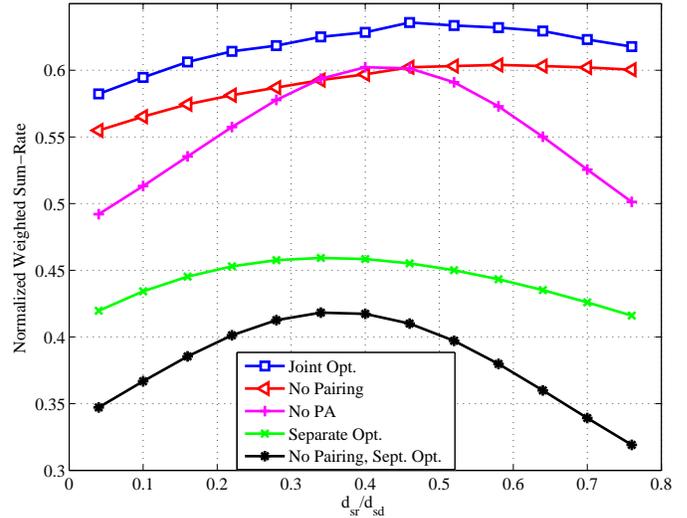}
\caption{Normalized weighted sum-rate vs. relay location; $K=4$, $\mb{w}=[.15,.15,.35,.35]$, $N=16$, and AF relaying.}  \label{rate_vs_dr_NeqW_AF}
\end{figure}

\section{Conclusion}\label{sec_conclusion}

We have studied the problem of jointly optimizing channel pairing, channel-user
assignment, and power allocation in a general single-relay multi-channel multi-user system. Although such joint optimization naturally leads to a mixed-integer programming formulation, we show that there is an efficient algorithm to find an optimal solution to our problem. The proposed approach transforms the original problem into a specially structured three-dimensional assignment problem, which not only preserves the binary constraints and strong Lagrange duality, but in some cases can also lead to polynomial-time computation complexity through careful choices of the optimization trajectory.  The proposed framework is applicable to a wide variety of scenarios.  The potentially significant improvement of system performance over suboptimal alternatives demonstrates the benefit of judicial design in such systems.

\appendices

\section{Derivation of Equation \eqref{eq_Ssrmnkopt}} \label{appendix_Smnk}

For notational simplicity, we drop all subscripts $m$, $n$, and $k$ from (\ref{Subproblem}).  We have the following maximization problem, which can be solved in the two cases below.
\begin{align}
&\max_{P^s,P^r} ~ \frac{w}{2} \tilde{\phi} \min\Big\{ \log\left(1+ \frac{aP^s}{\tilde{\phi}}\right), \nn \\
&  \log\left(1+ \frac{cP^s}{\tilde{\phi}}+ \frac{bP^r}{\tilde{\phi}}\right)  \Big\}
- (\l_s+\l_t) P^s - (\l_r+\l_t) P^r \label{Optimize} \\
s.t. &\quad P^s, P^r \geq 0 ~.  \nn
\end{align}

\subsection{Case One :  $a \leq c $}
In this case, the first term inside the $\min$ function in (\ref{Optimize}) is always smaller than the second term.  Hence, (\ref{Optimize}) is reduced to
\begin{align}
&\max_{P^s,P^r} ~ \frac{w}{2} \tilde{\phi} \log\left(1+ \frac{aP^s}{\tilde{\phi}}\right) - (\l_s+\l_t) P^s - (\l_r+\l_t) P^r \nn \\
s.t. &\quad  P^s, P^r \geq 0.  \label{Opt_case1}
\end{align}
Then, the optimal solutions from water-filling are obtained as
\beq
{P^s}^* = \left(\frac{w}{2(\l_s+\l_t)\ln 2} - \frac{1}{a} \right)^+ \tilde{\phi}, \quad {P^r}^* = 0 \label{OptSol1}
\eeq

\subsection{Case Two :  $a >  c $}
For this more complicated case, we propose the following solution.  We inspect the two possible outcomes in comparing the first and second terms in the $\min$ function in (\ref{Optimize}) \emph{at optimality}. Two separate maximization of (\ref{Optimize}) are performed under the constraint of either outcome.  Then, the optimal $(P^s, P^r)$ is given by the better of these two solutions.

\subsubsection{ Assumption 1: $a{P^s}^* \leq b{P^r}^* + c{P^s}^*$ }
Under this assumption, we have $b>0$ and the following optimization problem:
\begin{align}
&\max_{P^s,P^r} ~ \frac{w}{2} \tilde{\phi} \log  \left(1+ \frac{a P^s}{\tilde{\phi}}\right) - (\l_s+\l_t) P^s - (\l_r+\l_t) P^r \nn\\
s.t. &\quad  (i)~ P^s, P^r \geq 0  \nn \\
&\quad (ii)~ aP^s \leq bP^r + cP^s ~.  \label{Opt_case2_Assum1}
\end{align}
It has two possible solutions from the KKT conditions. One is obtained when the Lagrange multiplier corresponding to constraint (ii) is zero and the constraint is strictly satisfied. This implies that
\beq
{P^s}^* = \left(\frac{w}{2(\l_s+\l_t)\ln 2} - \frac{1}{a} \right)^+ \tilde{\phi}, \quad  {P^r}^* = 0 ~. \label{OptSol2_1}
\eeq
However, this solution contradicts with the assumption that (ii) is strictly satisfied. The other, correct solution occurs at the border $P^r = \frac{a-c}{b} P^s$.  By inserting this into the objective function, we have
\begin{align}
&{P^s}^* = \left(\frac{w b}{2 b(\l_s+\l_t)\ln 2 +2(a-c)(\l_r+\l_t)\ln 2} - \frac{1}{a} \right)^+\tilde{\phi} ~, \label{OptSol2_2}\\
&{P^r}^* = \frac{a-c}{b}{P^s}^* ~. \nn
\end{align}

\subsubsection{ Assumption 2: $a{P^s}^* \geq b{P^r}^* + c{P^s}^*$ }
Under this assumption, we have the following optimization problem:
\begin{align}
&\max_{P^s,P^r} ~ \frac{w}{2} \tilde{\phi} \log  \left(1+ \frac{c P^s}{\tilde{\phi}}+ \frac{b P^r}{\tilde{\phi}}\right) - \nn \\
& \quad \quad \quad (\l_s+\l_t) P^s - (\l_r+\l_t) P^r \nn\\
 s.t. &\quad (i)~ P^s, P^r \geq 0  \nn \\
 &\quad (ii)~ aP^s \geq bP^r + cP^s  \label{objec_Assum2}
\end{align}
From the KKT conditions, at optimality, either $b {P^r}^* = (a-c) {P^s}^*$, or the Lagrange multiplier corresponding to constraint (ii) is zero and the constraint is strictly satisfied.

In the former case, $b>0$ since $a \neq c$.  Furthermore, since the first and second terms in the $\min$ function in (\ref{Optimize}) are the same, we obtain the same solution as in \eqref{OptSol2_2}.

In the latter case, we define two new variables  $V^s=(\l_s + \l_t)P^s $ and $V^r=(\l_r + \l_t)P^r $. Substituting them into the objective of \eqref{objec_Assum2}, we have
\beq
\max_{V^s,V^r} \frac{w}{2} \tilde{\phi} \log\left(1+ \frac{c V^s}{\tilde{\phi}(\l_s+\l_t)} + \frac{b V^r}{\tilde{\phi}(\l_r+\l_t)}\right) - V^s - V^r ~. \label{Opt_V}
\eeq
The solution depends on the relation between $\frac{c}{\l_s+\l_t}$ and $\frac{b}{\l_r+\l_t}$:
\begin{itemize}

\item If $\frac{c}{\l_s+\l_t} > \frac{b}{\l_r+\l_t}$,
then we have ${V^r}^* = 0$, since otherwise a better solution to \eqref{Opt_V} would be $(V^s = {V^s}^* + {V^r}^*, V^r = 0)$.  Substituting ${V^r}^* = 0$ into \eqref{Opt_V}, we have ${V^s}^* = \left[\frac{w}{2\ln 2} - \frac{\l_s+\l_t}{c}\right]^+ \tilde{\phi}$.

\item If $\frac{c}{\l_s+\l_t} = \frac{b}{\l_r+\l_t}$,  \eqref{Opt_V} is a function of $(V^s+V^r)$ only, and ${V^r}^* = 0$ is a maximizer. Hence, again we have ${V^s}^* = \left[\frac{w}{2\ln 2} - \frac{\l_s+\l_t}{c}\right]^+ \tilde{\phi}$.

\item If $\frac{c}{\l_s+\l_t} < \frac{b}{\l_r+\l_t}$, similarly we have ${V^s}^* =0$.  However, this together with our assumption that constraint (ii) of \eqref{objec_Assum2} is strictly satisfied, i.e., $b{P^r}^* < (a-c){P^s}^*$, implies that ${V^r}^* < 0 $, which is not a feasible solution. Therefore, in this case at optimality the condition $b{P^r}^* = (a-c){P^s}^*$ prevails.

\end{itemize}

\section{Proof of Lemma \ref{lemma1}} \label{appendix1}

Given any $\tilde{\mathbf{\Phi}}=[\tilde{\phi}_{mnk}]_{N\times N\times K}$ with $0\leq \tilde{\phi}_{mnk} \leq 1$ and satisfying  \eqref{eq_tphi_sum}, let $x_{mn} = \sum_{k=1}^K \tilde{\phi}_{mnk}$. From $\sum_{m=1}^N\sum_{k=1}^K \tilde{\phi}_{mnk}=1$, we have $\sum_{m=1}^N x_{mn}=1$. Similarly, we have  $\sum_{n=1}^N x_{mn}=1$.  Hence $0 \leq x_{mn}\leq 1$. Then, $y_k^{mn}$ can be constructed as
\beq
y_k^{mn} = \begin{cases} \tilde{\phi}_{mnk}/x_{mn}, &x_{mn}>0 \\1/K, &x_{mn}=0 \end{cases} ~.
\eeq
Hence, $\sum_{k=1}^K y_k^{mn}=1$ and $0\leq y_k^{mn}\leq 1$. Note that $1/K$ above is arbitrarily chosen, and the mapping from $\tilde{\mathbf{\Phi}}$ to ($\mb{X}$, $\mb{y}^{mn}$) is one-to-many.

Given $\mb{X}$ and $\mb{y}^{mn}$ with $0 \leq x_{mn} \leq 1$ and $0\leq y_k^{mn} \leq 1$, satisfying
$\sum_{n=1}^N x_{mn} =1, \forall m$, $\sum_{m=1}^N x_{mn} =1, \forall n$,  and
$\sum_{k=1}^K y_k^{mn} =1, \forall m,n$, clearly $0\leq \tilde{\phi}_{mnk}=x_{mn}y_k^{mn} \leq 1$, and it is easy to verify that \eqref{eq_tphi_sum} is satisfied. This establishes the equivalence of $\tilde{\mathbf{\Phi}}$ and the proposed decomposition.

\section{Proof of Lemma \ref{lem_lambda_regions}}  \label{proof_lambda_regions}

We note that there exists at least one index vector $(m',n',k')$ such that $\phi^*_{m'n'k'}(\lvec^*) = 1$.  Furthermore, this chosen path must have non-degenerate user weight and channel gains so that the weighted rate function $w_{k'} R(m',n',k')$ is not uniformly zero, i.e., $w_{k'}>0$, $a_{m'}>0$, and $b_{n'k'}+c_{m'k'}>0$.

Suppose there exists $\phi^*_{m'n'k'}(\lvec^*) = 1$ such that $P_{m'n'k'}^{s*}(\lvec^*)$ is either $\big[\frac{w_{k'}}{\a (\l_s^*+\l_t^*)} - \frac{1}{a_{m'}} \big]^+$ or $\big[\frac{w_{k'}}{\a (\l_s^*+\l_t^*)} - \frac{1}{c_{m'k'}}\big]^+$.  In the former case, $a_{m'} \leq c_{m'k'}$, and the latter, $a_{m'} > c_{m'k'}$ and $\frac{c_{m'k'}}{\l_s^* + \l_t^*} \geq \frac{b_{n'k'}}{\l_r^* + \l_t^*}$. Furthermore, since $b_{n'k'}$ and $c_{m'k'}$ cannot both be zero in the latter case, $c_{m'k'} >0$.   Then we have
{\allowdisplaybreaks
    \begin{align}
    \min\{P_s,P_t\} &\geq \sum_{m,n,k} P_{mnk}^{s*}(\lvec^*) \geq P_{m'n'k'}^{s*}(\lvec^*) \nn\\
    &\geq \min\Big\{ \big[\frac{w_{k'}}{\a (\l_s^*+\l_t^*)} - \frac{1}{a_{m'}} \big]^+ , \nn \\ &\big[\frac{w_{k'}}{\a (\l_s^*+\l_t^*)} - \frac{1}{c_{m'k'}}\big]^+ \Big\} \nn\\
    &\geq \frac{w_{k'}}{\a (\l_s^*+\l_t^*)} - \frac{1}{\min\{ a_{m'}, c_{m'k'} \}} ~.
    \end{align}}%
    Hence,
    $\l_s^*+\l_t^* \geq \frac{w_{k'}\min\{ a_{m'}, c_{m'k'} \}}{\a (\min\{ a_{m'}, c_{m'k'} \}\min\{P_s,P_t\} + 1) }$,
    so $\lvec^* \in \mathcal{R}_1$.

Otherwise, for all $\phi^*_{mnk}(\lvec^*) = 1$, we have $a_{m} > c_{mk}$ and $P_{mnk}^{s*}(\lvec^*) = \big[ \frac{w_{k} b_{nk}}{\a (b_{nk} (\l_s^* + \l_t^*) + (a_{m} - c_{mk}) (\l_r^* + \l_t^*))} - \frac{1}{a_{m}} \big]^+$. We proceed with the following cases:

\begin{itemize}
\item If there exists $\phi^*_{m'n'k'}(\lvec^*) = 1$ such that $c_{m'k'} > 0$ and $\frac{c_{m'k'}}{\l_s^* + \l_t^*} < \frac{b_{n'k'}}{\l_r^* + \l_t^*}$, then we have
    \begin{align}
    &\min\{P_s,P_t\} \geq \sum_{m,n,k} P_{mnk}^{s*}(\lvec^*) \geq P_{m'n'k'}^{s*}(\lvec^*) \nn\\
    &\geq \frac{w_{k'} b_{n'k'}}{\a (b_{n'k'} (\l_s^* + \l_t^*) + (a_{m'} - c_{m'k'}) (\l_r^* + \l_t^*))} - \frac{1}{a_{m'}} ~, \nn\\
    &\geq \frac{w_{k'} b_{n'k'}}{\a (b_{n'k'} (\l_s^* + \l_t^*) + \frac{b_{n'k'}(a_{m'} - c_{m'k'})}{c_{m'k'}} (\l_s^* + \l_t^*))} - \frac{1}{a_{m'}} ~,
    \end{align}
    which implies that
    $\l_s^* + \l_t^* \geq \frac{w_{k'}c_{m'k'}}{\a (a_{m'}\min\{P_s,P_t\} + 1)}$,
    and hence $\lvec^* \in \mathcal{R}_1$.

\item Else, if there exists $\phi^*_{m'n'k'}(\lvec^*) = 1$ such that $\frac{c_{m'k'}}{\l_s^* + \l_t^*} \geq \frac{b_{n'k'}}{\l_r^* + \l_t^*}$, then $c_{m'k'} > 0$ since $b_{n'k'}$ and $c_{m'k'}$ cannot both be zero. In this case, the third expression in \eqref{eq_Ssrmnkopt} applies to the path $(m',n',k')$. Since $\big[ \frac{w_{k'} b_{n'k'}}{\a (b_{n'k'} (\l_s^* + \l_t^*) + (a_{m'} - c_{m'k'}) (\l_r^* + \l_t^*))} - \frac{1}{a_{m'}}\big]^+$ is the optimal power allocation for this path, we have
    \begin{align}
    \mathcal{L}_{m'n'k'} &(1, P^{s1}_{m'n'k'},P^{r1}_{m'n'k'},\lvec^*) \geq \nn \\ &\mathcal{L}_{m'n'k'} (1,P^{s2}_{m'n'k'},P^{r2}_{m'n'k'},\lvec^*) ~,
    \end{align}
    where
    \begin{align}
    P^{s1}_{m'n'k'} \defeq &\Big[ \frac{w_{k'} b_{n'k'}}{\a (b_{n'k'} (\l_s^* + \l_t^*) + (a_{m'} - c_{m'k'}) (\l_r^* + \l_t^*))} - \nn \\ &\frac{1}{a_{m'}}\Big]^+ \nn\\
    P^{s2}_{m'n'k'} \defeq &\Big[ \frac{w_{k'}}{\a (\l_s^* + \l_t^*)} - \frac{1}{c_{m'k'}} \Big]^+
    \end{align}
    correspond to the cases $a_{m'} P^{s1}_{m'n'k'} = b_{n'k'}P^{r1}_{m'n'k'} + c_{m'k'}P^{s1}_{m'n'k'}$ and $a_{m'} P^{s2}_{m'n'k'} > b_{n'k'}P^{r2}_{m'n'k'} + c_{m'k'}P^{s2}_{m'n'k'}$, respectively. This implies that \eqref{long_formula}.

    \begin{figure*}[!t]

    \begin{align}
    \frac{w_{k'}}{2} \log(1+a_{m'} P^{s1}_{m'n'k'}) - (\l_s^* + \l_t^*)
    (1+\frac{a_{m'} -c_{m'k'}}{b_{n'k'}}) P^{s1}_{m'n'k'} &\geq \frac{w_{k'}}{2} \log(1+c_{m'k'} P^{s2}_{m'n'k'}) - (\l_s^* + \l_t^*)P^{s2}_{m'n'k'} \nn\\
    \frac{w_{k'}}{2} \log(1+a_{m'} P^{s1}_{m'n'k'}) + (\l_s^* + \l_t^*)P^{s2}_{m'n'k'} &\geq \frac{w_{k'}}{2} \log(1+c_{m'k'} P^{s2}_{m'n'k'}) \nn\\
    \frac{w_{k'}}{2} \log(1+a_{m'} \min\{P_s,P_t\}) + \frac{w_{k'}}{\alpha} &\geq \frac{w_{k'}}{2} \log(1+c_{m'k'} P^{s2}_{m'n'k'}) \nn\\
    4+4a_{m'} \min\{P_s,P_t\} &\geq 1+c_{m'k'} P^{s2}_{m'n'k'} \nn\\
    4+4a_{m'} \min\{P_s,P_t\} &\geq \frac{w_{k'} c_{m'k'}}{\a (\l_s^* + \l_t^*)} ~. \label{long_formula}
    \end{align}
     \hrulefill
    \end{figure*}
    Hence, $$\l_s^* + \l_t^* \geq \frac{w_{k'} c_{m'k'}}{4\a (a_{m'} \min\{P_s,P_t\} + 1)},$$ so $\lvec^* \in \mathcal{R}_1$.

\item Else, the only scenario left is when $c_{mk} = 0$ for $\{m,k: \phi^*_{mnk}(\lvec^*) = 1\}$. In this case, there exists $b_{nk}>0$, since otherwise the achieved sum-rate is uniformly zero.  Then for any $(m',n',k')$ such that $\phi^*_{m'n'k'}(\lvec^*) = 1$,
    \begin{align}
    \min&\{P_s,P_t\} \geq \sum_{m,n,k} P_{mnk}^{s*}(\lvec^*) \geq P_{m'n'k'}^{s*}(\lvec^*) \nn\\
    &\geq \frac{w_{k'} b_{n'k'}}{\a (b_{n'k'} (\l_s^* + \l_t^*) + a_{m'} (\l_r^* + \l_t^*))} - \frac{1}{a_{m'}} ~,
    \end{align}
    which implies that
    $(\l_s^* + \l_t^*) + \frac{a_{m'}}{b_{n'k'}} (\l_r^* + \l_t^*) \geq \frac{w_{k'}}{\a (\min\{P_s,P_t\} + \frac{1}{a_{m'}})}$. Considering the extreme case for the slope and intercept of this linear inequality, we have $\lvec^* \in \mathcal{R}_2$.

\end{itemize}

\section{Proof of Lemma \ref{lem_ub_lambda}} \label{proof_ub_lambda}

To find an upper bound for $\|\lvec^*\|_2$, we consider the following cases for the activation pattern of the individual power constraints \eqref{eq_ind_power} and the total power constraint \eqref{eq_totalpower} at global optimum.
\begin{itemize}
\item{\textit{Neither constraint in \eqref{eq_ind_power} is active:}} In this case, \eqref{eq_totalpower} must be active, since otherwise there would be more power to increase the sum-rate. Thus, we have $\l_s^*=\l_r^*=0$ and
    \beq
    P_t = \sum_{m,n,k} P_{mnk}^{s*}(\lvec^*) + P_{mnk}^{r*}(\lvec^*)~.
    \eeq
    Since $\phi^*_{mnk}(\lvec^*) \leq 1$, substituting $\phi^*_{mnk}(\lvec^*)=1$ and $\l_s^*=\l_r^*=0$ into \eqref{eq_Ssrmnkopt}, and considering all possible scenarios of \eqref{eq_Ssrmnkopt}, we have
    \begin{align}
    P_t &\leq \sum_{m,n,k} \max\{ \big[\frac{w_{k}}{\a \l_t^*} - \frac{1}{a_{m}}\big]^+, \big[\frac{w_{k}}{\a \l_t^*} - \frac{1}{c_{mk}}\big]^+, \frac{w_{k}}{\a \l_t^*} \} \nn\\
    &= \sum_{m,n,k} \frac{w_{k}}{\a \l_t^*} = \frac{N^2}{\a \l_t^*}~.
    \end{align}
    Hence, we have $\l_t^* \leq \frac{N^2}{\a P_t}$,
    so that
    \beq
    \|\lvec^*\|_2 \leq \frac{N^2}{\a P_t} ~.
    \eeq
\item{\textit{Both constraints in \eqref{eq_ind_power} are active:}} We have
    \begin{align}
    P_s &= \sum_{m,n,k} P_{mnk}^{s*}(\lvec^*) ~, \nn\\
    P_r &= \sum_{m,n,k} P_{mnk}^{r*}(\lvec^*) ~.
    \end{align}
    Again, substituting $\phi^*_{mnk}(\lvec^*)=1$ into \eqref{eq_Ssrmnkopt}, and considering all possible scenarios, we conclude that
    {\allowdisplaybreaks
    \begin{align}
    & P_s \leq \sum_{m,n,k} \max\Big\{ \big[\frac{w_{k}}{\a (\l_s^*+\l_t^*)} - \frac{1}{a_{m}}\big]^+, \nn \\
    &\quad \big[\frac{w_{k}}{\a (\l_s^*+\l_t^*)} - \frac{1}{c_{mk}}\big]^+, \nn\\
    & \quad  \frac{1(a_m > c_{mk})w_{k}b_{nk}}{\a (b_{nk} (\l_s^*+\l_t^*) + (a_m-c_{mk})(\l_r^*+\l_t^*)) } \Big\} \nn\\
    &\leq \sum_{m,n,k} \frac{w_{k}}{\a (\l_s^*+\l_t^*)} = \frac{N^2}{\a (\l_s^*+\l_t^*)} ~, \label{eq_bound_ps}\\
    P_r &\leq \sum_{m,n,k}  \frac{1(a_m > c_{mk})w_{k}(a_m-c_{mk})}{\a (b_{nk} (\l_s^*+\l_t^*) + (a_m-c_{mk})(\l_r^*+\l_t^*)) } \nn\\
    &\leq \sum_{m,n,k} \frac{w_{k}}{\a (\l_s^*+\l_t^*)}  = \frac{N^2}{\a (\l_r^*+\l_t^*)} ~, \label{eq_bound_pr}
    \end{align}
    Hence, we have $\l_s^*+\l_t^* \leq \frac{N^2}{\a P_s}$, and $\l_r^*+\l_t^* \leq \frac{N^2}{\a P_r}$, so that
    \beq
    \|\lvec^*\|_2 \leq \frac{\sqrt{2} N^2}{\a \min \{P_s, P_r\}}~.  \label{eq_bound_pspr}
    \eeq
    }%
\item{\textit{Only one constraint in \eqref{eq_ind_power} is active:}} We have either $\l_r^*=0$ and \eqref{eq_bound_ps}, or $\l_s^*=0$ and \eqref{eq_bound_pr}.  For both cases, an upperbound for $\|\lvec^*\|_2$ is given by \eqref{eq_bound_pspr}.
\end{itemize}
Summarizing the three cases above, we have
\beq
\|\lvec^*\|_2 \leq \frac{\sqrt{2} N^2}{\a \min \{P_s, P_r, P_t\}} ~.  \label{eq_bound_psprpt}
\eeq

\section{Proof of Lemma \ref{lem_ub_theta}} \label{proof_ub_theta}

We first note that there must exist one path $(m,n,k)$ with non-degenerate user weight and channel gains so that the weighted rate function $w_{k} R(m,n,k)$ is not uniformly zero, i.e., $w_{k}>0$, $a_{m}>0$, and $b_{nk}+c_{mk}>0$.
In Algorithm \ref{alg_subgradient}, subgradient updating is performed over two cases, either there exists some $c_{mk}>0$ and $\lvec^{(l)} \in \mathcal{R}_1$, or $c_{mk}=0$ for all $m$ and $k$ and $\lvec^{(l)} \in \mathcal{R}_2$.

In the former case, let $\epsilon_1 = \frac{\displaystyle \min_{\{k: w_k>0\}} w_k  \min \{ \min_{\{m: a_m>0\}} a_m, \min_{\{m,k: c_{mk}>0\}} c_{mk} \}}{\displaystyle 4\a (\max_m a_m \min\{P_s,P_t\} + 1)}$.  The projected subgradient updating is performed over $\l_s^{(l)}+\l_r^{(l)} \geq \epsilon_1$.  Since there exists some $m$ and $k$ such that $w_k>0$, $a_m>0$, and $c_{mk}>0$, we have $\epsilon_1 > 0$.  From \eqref{eq_Ssrmnkopt}, we see that in all scenarios
\begin{align}
P^{s*}_{mnk}(\lvec^{(l)}) &\leq \frac{w_k}{\a (\l_s^{(l)}+\l_r^{(l)})} \leq \frac{w_k}{\a \epsilon_1} <\infty ~, \label{ps_1}\\
P^{r*}_{mnk}(\lvec^{(l)}) &\leq \frac{a_m - c_{mk}}{b_{nk}} P^{s*}_{mnk}(\lvec^{(l)}) \leq \frac{w_k (a_m - c_{mk})}{\a \epsilon_1 b_{nk}} <\infty ~. \label{pr_1}
\end{align}
The second inequality above hold since $P^{r*}_{mnk}(\lvec^{(l)}) = 0$ when $b_{nk}=0$.

In the latter case, let $\epsilon_2 = \frac{\displaystyle \min_{\{k: w_k>0\}} w_k}{\displaystyle \a \big(\min\{P_s,P_t\} + \frac{1}{\displaystyle \min_{\{m: a_m>0\}} a_m}\big) }$. The projected subgradient updating is performed over $(\l_s + \l_t) + \frac{\displaystyle \min_{\{m: a_m>0\}} a_m }{\displaystyle \max_{n,k} b_{nk}} (\l_r + \l_t) \geq \epsilon_2$.  Since there exists some $m$ and $k$ such that $w_k>0$, $a_m>0$, we have $\epsilon_2>0$. Furthermore, $b_{nk}>0$ since   From \eqref{eq_Ssrmnkopt}, we see that in all scenarios
\begin{align}
P^{s*}_{mnk}(\lvec^{(l)}) &\leq \frac{w_k}{\a (\l_s^{(l)}+\frac{a_m}{b_{nk}} \l_r^{(l)})} \leq \frac{w_k}{\a \epsilon_2} <\infty ~, \label{ps_2}\\
P^{r*}_{mnk}(\lvec^{(l)}) &\leq \frac{a_m - c_{mk}}{b_{nk}} P^{s*}_{mnk}(\lvec^{(l)}) \leq \frac{w_k (a_m - c_{mk})}{\a \gamma b_{nk}} <\infty ~. \label{pr_2}
\end{align}
The second inequality above hold since $P^{r*}_{mnk}(\lvec^{(l)}) = 0$ when $b_{nk}=0$.

We note that the bounds in \eqref{ps_1}-\eqref{pr_2} are not functions of $m$ and $n$. Substituting the bounds of either of these two cases into \eqref{subgradient}, and noting that only $N^2$ paths are chosen, we have $\|\thetavec(\lvec)\|_2 = O(N^2)$.

\bibliographystyle{IEEEtran} 
\bibliography{IEEEabrv,master}

\begin{IEEEbiography}{Mahdi Hajiaghayi}
(S'10) received the B.Sc.~degrees in electrical engineering and petroleum engineering simultaneously from Sharif University of Technology, Tehran, Iran, in 2005, and the M.Sc.~degree in electrical and computer engineering from the University of Alberta, Edmonton, Alberta, Canada, in 2007. He is currently pursuing the Ph.D.~degree in electrical and computer engineering at the University of Toronto, Toronto, Ontario, Canada. He won the Ontario Graduate Scholarship in 2010-2011.
\end{IEEEbiography}

\begin{IEEEbiography}{Min Dong}
(S'00-M'05-SM'09) received the B.Eng. degree from Tsinghua University, Beijing, China, in 1998, and the Ph.D. degree in Electrical and Computer Engineering with minor in Applied Mathematics from Cornell University, Ithaca, NY, in 2004. During 2004-2008, she was with Corporate Research and Development, Qualcomm Inc., San Diego, CA. Since July 2008, she has been with the Faculty of Engineering and Applied Science at the University of Ontario Institute of Technology, Ontario, Canada, where she is currently an Assistant Professor. She also holds a status-only Assistant Professor appointment in the Department of Electrical and Computer Engineering at the University of Toronto. She received the Ontario MEDI Early Researcher Award in 2012, and the 2004 IEEE Signal Processing Society Best Paper Award. She currently serves as an Associate Editor for the IEEE Transactions on Signal Processing and IEEE Signal Processing Letter.
\end{IEEEbiography}

\begin{IEEEbiography}{Ben Liang}
(S'94-M'01-SM'06) received honors simultaneous B.Sc.~(valedictorian) and M.Sc.~degrees in electrical engineering from Polytechnic University in Brooklyn, New York, in 1997 and the Ph.D.~degree in electrical engineering with computer science minor from Cornell University in Ithaca, New York, in 2001. In the 2001 - 2002 academic year, he was a visiting lecturer and post-doctoral research associate at Cornell University. He joined the Department of Electrical and Computer Engineering at the University of Toronto in 2002, where he is now a Professor. He received an Intel Foundation Graduate Fellowship in 2000 and an Ontario MRI Early Researcher Award in 2007. He was a co-author of the Best Paper Award in the IFIP Networking conference in 2005 and a Best Paper Finalist in IEEE INFOCOM in 2010. He is an editor for the IEEE Transactions on Wireless Communications and an associate editor for the Wiley Security and Communication Networks journal.
\end{IEEEbiography}

\end{document}